\documentclass[11pt]{article}
\usepackage[colorlinks=true,linkcolor=black,citecolor=blue]{hyperref}
\usepackage{forest}
\usetikzlibrary{decorations.pathreplacing,calc}
\forestset{solid nodes/.style={for tree={circle,draw,inner sep=1,fill=green}},
            dir/.style={for tree={grow=#1}},
            leaf/.style={label=$#1$},
            mytree/.style={solid nodes, for tree={grow=south,s sep=1cm}}}
\usepackage[utf8]{inputenc}

\usepackage{fullpage}
\usepackage[margin = 2.5cm]{geometry}

\usepackage{amsmath, amssymb, amsthm}

\usepackage{graphicx}
\usepackage{color}
\usepackage{tikz-qtree}  
\usepackage{tcolorbox}
\usepackage{forest}

\usepackage[round]{natbib}

\usepackage{nicefrac}
\usepackage{textcomp}
\usepackage{enumerate}

\newtheorem{theorem}{Theorem}[section]
\newtheorem{claim}[theorem]{Claim}
\newtheorem{corollary}[theorem]{Corollary}
\newtheorem{lemma}[theorem]{Lemma}
\newtheorem{definition}[theorem]{Definition}

\newcommand{\E}{\mathbf{E}}

\newcommand{\bbN}{\mathbb{N}}

\newcommand{\bbR}{\mathbb{R}}

\newcommand{\calA}{\mathcal{A}}
\newcommand{\calB}{\mathcal{B}}

\newcommand{\calD}{\mathcal{D}}

\newcommand{\calI}{\mathcal{I}}

\newcommand{\calP}{\mathcal{P}}
\newcommand{\calQ}{\mathcal{Q}}

\definecolor{Blue}{HTML}{2D2F92}

\DeclareMathOperator{\val}{val}
\DeclareMathOperator{\opt}{opt}
\DeclareMathOperator{\mc}{mc}
\DeclareMathOperator{\dist}{dist}
\DeclareMathOperator{\cousins}{cousins}
\DeclareMathOperator{\triplet}{triplet}

\DeclareMathOperator{\lca}{LCA}

\DeclareMathOperator{\Avg}{Avg}

\DeclareMathOperator{\weight}{weight}

\DeclareMathOperator{\order}{order}
\DeclareMathOperator{\clr}{color}

\DeclareMathOperator{\argmin}{argmin}

\DeclareMathOperator{\child}{child}

\DeclareMathOperator{\ONE}{\mathbf{1}}

\usepackage{mdframed}
\newtheorem{result}{Theorem}

\title{Triplet Reconstruction and all other Phylogenetic CSPs\\ are Approximation Resistant}
\author{Vaggos Chatziafratis$^1$ \and Konstantin Makarychev$^2$}
\date{%
    $^1$University of California, Santa Cruz\\%
    $^2$Northwestern University\\[2ex]%
}
\begin{document}
\maketitle
\thispagestyle{empty}
\begin{abstract}
We study the natural problem of Triplet Reconstruction (also known as Rooted Triplets Consistency or Triplet Clustering), originally motivated by applications in computational biology and relational databases~\citep*{aho1981inferring}: given $n$ datapoints, we want to embed them onto the $n$ leaves of a rooted binary tree (also known as a hierarchical clustering, or ultrametric embedding) such that a given set of $m$ \textit{triplet} constraints is satisfied. A triplet constraint $ij|k$ for points $i,j,k$ indicates that ``$i, j$ are more closely related to each other than to $k$,” (in terms of distances $d(i,j)\le d(i,k)$ and $d(i,j)\le d(j,k)$) and we say that a tree satisfies the triplet $ij|k$ if the distance in the tree between $i,j$ is smaller than the distance between $i,k$ (or $j,k$). Among all possible trees with $n$ leaves, can we efficiently find one that satisfies a large fraction of the $m$ given triplets? 

~\cite{aho1981inferring} studied the decision version and gave an elegant polynomial-time algorithm that determines whether or not there exists a tree that satisfies all of the $m$ constraints. Moreover, it is straightforward to see that a random binary tree achieves a constant $\tfrac13$-approximation, since there are only 3 distinct triplets $ij|k,ik|j,jk|i$ (each will be satisfied w.p. $\tfrac13$). Unfortunately, despite more than four decades of research by various communities, there is no better approximation algorithm for this basic Triplet Reconstruction problem.

Our main theorem---which captures Triplet Reconstruction as a special case---is a general hardness of approximation result about Constraint Satisfaction Problems (CSPs) over \textit{infinite} domains (CSPs where instead of boolean values $\{0,1\}$ or a fixed-size domain, the variables can be mapped to any of the $n$ leaves of a tree). Specifically, we prove that assuming the Unique Games Conjecture~\citep{khot2002power}, Triplet Reconstruction and more generally, \textit{every} Constraint Satisfaction Problem (CSP) over \textit{hierarchies} is approximation resistant, i.e., there is no polynomial-time algorithm that does asymptotically better than a \textit{biased} random assignment. 

Our result settles the approximability not only for Triplet Reconstruction, but for many interesting problems that have been studied by various scientific communities such as the popular Quartet Reconstruction and Subtree/Supertree Aggregation Problems. More broadly, our result significantly extends the list of
approximation resistant predicates by pointing to a large new family of hard problems over hierarchies. Our main theorem is a generalization of~\citet*{GHMRC11}, who showed that \textit{ordering} CSPs (CSPs over permutations of $n$ elements, e.g., Max Acyclic Subgraph, Betweenness, Non-Betweenness) are approximation resistant. The main challenge in our analyses stems from the fact that trees have \textit{topology} (in contrast to permutations and ordering CSPs) and it is the tree topology that determines whether a given constraint on the variables is satisfied or not. As a byproduct, we also present some of the first CSPs where their approximation resistance is proved against biased random assignments, instead of uniformly random assignments.

\end{abstract}


\newpage
\tableofcontents
\thispagestyle{empty}
\newpage
\setcounter{page}{1}
\section{Introduction}

The algorithmic task of constructing \textit{hierarchical} representations of data has been studied by various communities over many decades with applications ranging from statistics~\citep*{ward1963hierarchical,hastie2009elements} and databases~\citep*{aho1981inferring} to the analysis of complex networks, such as the Internet or social networks~\citep*{ravasz2003hierarchical,clauset2008hierarchical}, and more recently, to machine learning, where hierarchical embeddings have proven useful for understanding text, images, graphs and multi-relational data~\citep{nickel2017poincare}. The reason why they are so ubiquitous is that many real data sets stemming from nature or society are organized according to a latent hierarchy~\citep*{ravasz2003hierarchical}. Interestingly, many relevant questions and algorithmic ideas originated in the field of Taxonomy and Phylogenetics~\citep*{sneath1973numerical,eisen1998cluster,felsenstein2004inferring} with the goal of classifying all living and extinct organisms into the Tree of Life.


The easiest way to visualize such hierarchical representations for a given data set is by using a dendrogram, also known as a \textit{Hierarchical Clustering}. Hierarchical clustering is an embedding of the data set to a tree, often depicted as a rooted binary tree whose leaves are in one-to-one correspondence with the points in the data set, see Figure~\ref{fig:triplets}. The hierarchical clustering tree shows the recursive partitioning of the data set into successively smaller and smaller clusters. Observe that all data points are clustered together at the root, but eventually they get separated at the leaves (internal nodes correspond to intermediate subclusters formed by descendant leaves). 

\begin{figure}[ht]
    \centering
    \includegraphics[scale=0.87]{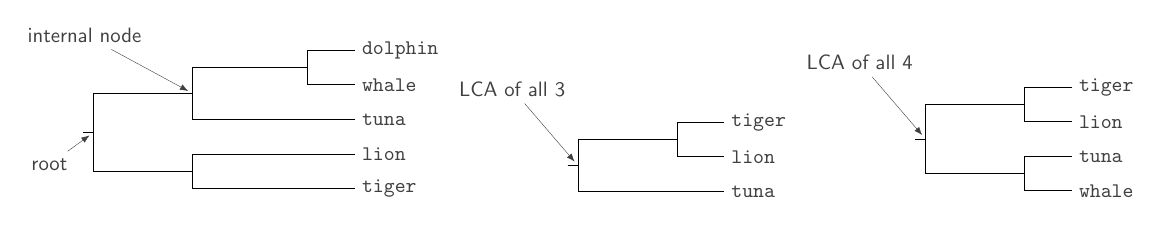}
    \caption{Hierarchical Clustering on 5 points (Left), a triplet constraint (Middle) and a quartet constraint (Right) satisfied by hierarchical clustering. The internal node shown corresponds to the subcluster $\texttt{\{whale, dolphin, tuna\}}$. The basic constituent of a hierarchy is a \textit{triplet comparison} or \textit{triplet constraint}, e.g., $\{\texttt{\{lion, tiger\}|\{tuna\}}\}$ indicates the closest pair among the 3. Formally, the Lowest Common Ancestor (LCA) of $\texttt{\{lion, tiger\}}$ is a descendant of the LCA of all 3. Another type of a more complicated comparison is among 4 points, e.g., the \textit{quartet} comparison $\{\texttt{\{lion, tiger\}|\{tuna, whale\}}\}$ prescribes what's the correct split. We can see that the hierarchical clustering satisfies the shown triplet and  quartet constraint (it also satisfies triplets $\{\texttt{\{whale, dolphin\}|\{tuna\}} \}$, $\{\texttt{\{whale, dolphin\}|\{lion\}} \}$, $\{\texttt{\{whale, tuna\}|\{tiger\}} \}$, and quartets $\{\texttt{\{whale, dolphin\}|\{lion, tiger\}}\}$, $\{\texttt{\{\{whale, dolphin\}|\{tuna\}\}|\{lion\}}\}$.}
    \label{fig:triplets}
\end{figure}

In contrast to ``flat'' clustering techniques like $k$-means/$k$-median which cannot capture fine-grained relationships among points or groups of points, hierarchical clustering reveals the structure of a data set at multiple levels of granularity simultaneously. For example, consider \textit{triplet} queries of the form ``\textit{Among 3 items $i,j,k$, which two are most closely related?}'';  a quick look at the hierarchical clustering (see Fig.~\ref{fig:triplets})  immediately reveals the answer by locating the 3 leaves $i,j,k$ and noticing which of the 3 gets separated first from the other two. Answering such triplet comparisons is easy for humans which makes them popular in metric learning and crowdsourcing settings~\citep*{tamuz2011adaptively,vikram2016interactive,emamjomeh2018adaptive}, and understanding how to accurately aggregate a large collection of such triplet queries into a hierarchical clustering was the primary motivation of our work. As we will see, studying triplets will lead us to interesting connections with hardness of approximation and approximation resistant predicates.

In this paper,  we study the approximability of a large class of Constraint Satisfaction Problems (CSPs) over \textit{hierarchies}, i.e., trees, which have been studied in various communities, including in databases~\citep*{aho1981inferring}, in logic and algebra~\citep*{bodirsky2010complexity}, in computational biology~\citep*{felsenstein2004inferring,byrka2010new} and in theoretical computer science~\citep*{jiang1998orchestrating,snir2008quartets,brodal2013efficient,alon2014compatibility}. The input is a collection of (potentially inconsistent) local relationships between $k$ items of a ground set (with total size $n$), and we are asked to find the hierarchical clustering that \textit{maximizes} agreement with the input. Such local relationships can take the form of triplet or quartet constraints (or even quintuples etc.), and more generally, they can be a $k$-arity constraint on $k$ data points which prescribes how the $k$ data points should be split by the final hierarchy. For the most common examples of triplet and quartet constraints, please see Figure~\ref{fig:triplets}. 

For readers familiar with Correlation Clustering~\citep*{bansal2004correlation},
we should note here that it is different in at least three important ways: first, in correlation clustering the desired output is a (flat) partition of the data points (whereas in hierarchical clustering we want a mapping to the $n$ leaves of a tree), second, constraints in correlation clustering are between pairs of points (whereas in hierarchical clustering the input specifies constraints on triplets, quartets etc.) and third, there are technical differences (as we show) in terms of their behavior with respect to approximation resistance.

\subsection{Our Contributions} 
We revisit several old questions in Hierarchical Clustering and CSPs over \textit{Infinite}-Domains and  prove \textit{tight} upper and lower bounds under the Unique Games Conjecture~\citep{khot2002power}, thus settling the approximability of a large class of hierarchical reconstruction problems. Interestingly, we extend the notion of \textit{approximation resistant} CSPs~\citep{haastad2001some} to allow for \textit{biased} random assignments (instead of uniform random assignments), and our main hardness result for CSPs over trees holds under this extended definition, which could be of independent interest. As far as we know, our results provide the first approximation resistant CSPs where the optimal approximation threshold is achieved by a \textit{non-uniform} random assignment.

Recall that for CSPs over infinite domains, the variables are not boolean and instead of taking values $\{0,1\}$ (or in a fixed-size domain), they are allowed to be mapped to infinite domains. Prominent examples include Correlation Clustering~\citep{bansal2004correlation}
and \textit{Ordering CSPs}, i.e., CSPs over permutations of $n$ elements, such as Max Acyclic Subgraph, Betweenness, Non-Betweenness~\citep*{GHMRC11}; for our case, the infinite domain corresponds to the $n$ leaves of a hierarchical tree which of course grows as the number $n$ of data points grows. In fact, our results generalize the hardness results of \cite{GHMRC11}, since a permutation corresponds (in a formal sense) to a special case of hierarchical clustering (because we consider ordered trees). At a high-level, the main challenge in our problems comes from the fact that CSPs over trees depend on the tree's \textit{topology} and whether a given constraint is satisfied or not is determined by how and in what order exactly various data points got split at intermediate nodes. Observe that this is irrelevant for correlation clustering and for permutations.
Specifically, we settle the approximability of 
Triplet Reconstruction, Quartet Reconstruction, and General CSPs over Trees.


\textbf{Triplet Reconstruction (also Rooted Triplets Consistency):} Aggregating triplets into a hierarchical clustering was originally asked in the context of relational databases by~\citet{aho1981inferring}. A triplet constraint $ij|k$ indicates that ``items $i,j$ are more similar to each other than to $k$.'' 
Given $m$ triplets, we would like to construct a hierarchical clustering on the $n$ items, that satisfies as many constraints as possible, i.e., $k$ is split first from $i,j$ (see also Fig.~\ref{fig:triplets}). 

\textbf{Quartet Reconstruction:} When constraints are on 4 points $a,b,c,d$, they are called ``quartet'' constraints.  The task is to find a (rooted or unrooted) tree that satisfies as many of the  quartet constraints as possible (Fig.~\ref{fig:triplets}). A special case of Quartet Reconstruction is the popular ``Unrooted Quartet Consistency'' problem in computational biology~\citep{steel1992complexity,jiang1998orchestrating,ben1998constructing,felsenstein2004inferring,snir2008quartets,snir2012quartet,alon2014compatibility}.   


\textbf{General CSPs over Trees:} The previous two problems are only two special cases of general CSPs over trees. Specifically, Triplet Reconstruction is a CSP of arity 3 and Quartet Reconstruction is a CSP of arity 4. However, there is no reason to stop there; in fact, the algebraic and logic communities have extensively studied  what happens if we allow for trees with larger fan-out degree, or for conjunctions (logical $\land$) or disjunctions (logical $\lor$) between constraints,\footnote{For example, $ij|k \lor ik|j$ captures the \textit{forbidden} triplets problem, that forbids triplets $jk|i$ from the final hierarchy.} or for the constraints to be of arity $k$. In the algebraic and logic literature~\citep*{bodirsky2010complexity,bodirsky2012complexity}, such CSPs are termed \textit{Phylogenetic CSPs} due to their connections to popular ``Consensus Tree'' or ``Subtree/Supertree Aggregation'' methods in computational biology~\citep{adams1972consensus,steel1992complexity,sanderson1998phylogenetic,ng2000difficulty,jansson2016improved}.  

Before stating our general theorem, let us focus only on our first result about the Triplet Reconstruction problem and highlight its status in terms of approximability and hardness. Then, it will be much easier to understand our results for Quartet Reconstruction and for General CSPs over Trees (along with the main technical challenges).

\subsection{Result I: Beating Random is Hard for Triplet Reconstruction}

We will need the following simple definitions (for examples, see Fig~\ref{fig:triplets}):

\begin{definition}[Triplet]
A triplet $t$, denoted $t=ab|c$, is a rooted, unordered, binary tree on 3 leaves $a,b,c$. A rooted, unordered, binary tree $T$ (containing leaves $a,b,c$) is said to be \emph{consistent with $t$} (or $T$ \emph{satisfies} $t$), if the LCA($a,b$) in $T$ is a proper descendant of LCA($a,c$) in $T$. Otherwise, the triplet and the tree are \emph{inconsistent with each other} (or $T$ \emph{violates} $t$). In general, triplets can also have weights $\weight(ab|c)$.
\end{definition}

The natural optimization problem associated with Triplet Reconstruction is \textsc{MaxTriplets}:
\begin{definition}[\textsc{MaxTriplets} Problem]
Given a set $X$ of $n$ data points and $m$ triplets defined on data points from $X$, find the hierarchical clustering (i.e., the binary rooted tree) that is consistent with as many triplets as possible (per the definition above). 
\end{definition}

\textsc{MaxTriplets} is NP-hard in general instances, but \citet*{aho1981inferring} presented an algorithm for completely satisfiable instances  of \textsc{MaxTriplets}: This algorithm finds a tree that is consistent with all given triplets, if such tree exists, otherwise it halts and declares that the triplets are conflicting and no tree can satisfy all of the triplets.


The following trivial algorithm achieves a $\tfrac13$-approximation: ``output a \textit{uniformly random} tree on the $n$ data points.'' Observe that for 3 items $a,b,c$, there are only 3 distinct triplets---namely $ab|c, ac|b, bc|a$---and so with probability $\tfrac13$, the uniformly random tree will satisfy each of the input triplets. 
Surprisingly, despite four decades of research, this is currently the best known approximation for triplet reconstruction.
Our first result shows that being stuck at the trivial $\tfrac13$-approximation ratio is not a coincidence:

\begin{tcolorbox}
\begin{theorem}\label{th:intro_triplets} For every constant $\varepsilon>0$, it is UG-hard to distinguish instances of  {\sc{MaxTriplets}}, where a $(1-\varepsilon)$ fraction
of the triplets can be satisfied by a hierarchical clustering, from {\sc{MaxTriplets}} instances where at most a $(\tfrac13+ \varepsilon)$ fraction can be satisfied.
\end{theorem}
\end{tcolorbox}

Stated simply, we prove that if $\rho$ is the expected fraction of constraints satisfied by a uniformly random tree, then obtaining a $\rho'$-approximation for any constant $\rho'>\rho$ is UG-hard. In other words, we show that {\sc{MaxTriplets}} is \textit{approximation resistant}. Recall that a predicate is approximation resistant if it is NP-hard (or UG-hard in our case) to approximate the corresponding CSP significantly better than what is achieved by the trivial algorithm that picks an assignment uniformly at random. For example, 3SAT is approximation resistant~\citep{haastad2001some} and so is every ordering CSP such as Max Acyclic Subgraph, Betweenness, Non-Betweenness~\citep*{GHMRC11}. Prior to our work, the best known hardness of approximation
for {\sc{MaxTriplets}} was \nicefrac{2}{3} due to~\cite*{chatziafratis2021maximizing}.  


\subsection{Result II: From Triplets to Hardness of General CSPs over Trees}

Given our first result on the hardness of {\sc{MaxTriplets}}, it is natural to wonder what happens in terms of approximability if we increase the arity of the constraints from 3 to 4, i.e., what happens for Quartet Reconstruction and its associated optimization versions of {\sc{MaxQuartets}} (we defer the exact definitions for now, but hopefully the problem is clear). For {\sc{MaxQuartets}} even though there are results~\citep*{jiang1998orchestrating} that give a PTAS for very dense instances (density here implies that there is a quartet for every four data points, thus $m=\Omega(n^4)$), the approximability in the general case remained open:
\emph{How well can we approximate \textsc{MaxQuartets} in polynomial time?}
Again, for the most well-studied case of unrooted quartet reconstruction~\citep{jiang1998orchestrating,ben1998constructing,felsenstein2004inferring}, a trivial algorithm that outputs an unrooted tree uniformly at random is a constant approximation, and this has been the state-of-the-art in the worst-case for many decades. In light of our Theorem~\ref{th:intro_triplets}, we are able to settle the approximability for \textsc{MaxQuartets}, proving that this trivial algorithm is again optimal (under UGC) (see Appendix~\ref{sec:triplets-to-quartets}).

\paragraph{General CSPs over Trees.} Triplets and Quartets are two special cases of a more general family of CSPs over trees that are not well-understood from a theoretical perspective. Such general CSPs over trees are also studied in the algebraic and logic communities under the name \textit{Phylogenetic CSPs}~\citep{bodirsky2010complexity,bodirsky2012complexity,bodirsky2017complexity}, which will be borrowed here.\footnote{A technical comment is that our definition of Phylogenetic CSPs is slightly more general than the one in the logic community~\citep{bodirsky2017complexity}:  they only focused on unordered, rooted trees, whereas our results hold even for CSPs on \textit{ordered} trees (left and right children are distinguishable) and on unrooted trees. Ordered trees play an important role when mapping a hierarchy to a permutation on its leaves with specific structure~\citep{bar2001fast,geary2006succinct,jansson2007ultra} and in consensus methods~\citep{jansson2006algorithms,jansson2016improved}.} For formal definitions, please see Preliminaries (Section~\ref{sec:prelim}).


After seeing our hardness results for {\sc{MaxTriplets}} and 
{\sc{MaxQuarters}}, one may assume that the random assignment algorithm always gives the best possible approximation (ignoring $o(1)$ terms). However, as we discuss in Section~\ref{sec:random-assignment-def}, this is not the case. In fact, for some phylogenetic CSPs the uniform random assignment algorithm satisfies exponentially small in $k$ fraction of all constraints while other algorithms satisfy e.g., a constant (not depending on $k$) fraction of all constraints. It turns out that the best algorithms for arbitrary phylogenetic CSPs are \textit{biased} random assignment algorithms. We show the following result.

\begin{tcolorbox}
\begin{theorem}[Informal]\label{th:main}
 Assuming the UGC, \textit{every} Phylogenetic CSP is approximation resistant. Interestingly, this holds for a more general notion of approximation resistance, where \textbf{biased} random solutions are allowed (not just uniformly random outputs like in boolean CSPs and ordering CSPs).
\end{theorem}
\end{tcolorbox}

\paragraph{On Approximation Resistance.} The subject of approximation resistance is a fascinating topic in computation with a rich literature, and despite the intensive efforts to characterize the approximability of CSPs, it is not yet clear what properties characterize them in general. It is perhaps striking, but many CSPs are approximation resistant, and two fundamental examples are {\sc{Max3SAT, Max3LIN}}~\citep{haastad2001some}. In contrast, for arity 2,~\citet{haastad2005every} showed that no predicate that depends on two inputs (e.g., {\sc{MaxCUT}}) from an arbitrary finite domain can be approximation resistant. Investigating higher arity CSPs has also yielded interesting results: for arity 3, a precise classification of approximation resistant 3CSPs is known~\citep*{zwick1998approximation}, but for arity 4 and higher the situation is unclear~\citep*{hast2005beating,haastad2007approximation,austrin2009approximation,austrin2009randomly}. For example, Hast gave a characterization for 355 out of 400 different predicate types for binary 4CSPs. Moreover,~\cite{haastad2007approximation} showed, under  UGC, that a random $k$-ary predicate for large $k$ is approximation resistant. More recently,~\cite{guruswami2015towards} showed hardness for the family of symmetric CSPs (predicates whose set of accepting strings is permutation invariant).

\paragraph{Ordering and Ordinary CSPs} Beyond the above finite-domain CSPs, approximation resistance has been studied for \textit{infinite-domain} CSPs (or ``growing'' domain CSPs). Several prominent such examples that were shown approximation resistant include  Maximum Acyclic Subgraph~\citep*{guruswami2008beating}, Betweenness, Non-Betweenness, Cyclic Ordering~\citep*{charikar2009every} and in fact, any other \textit{ordering} CSP (CSPs over permutations of $n$ elements) is approximation resistant~\citep*{GHMRC11}. Each predicate or payoff function of an ordering CSP depends on the ordering of variables on a line.

In this paper, we will use not only phylogenetic CSPs but also CSPs with finite alphabet and ordering CSPs. To distinguish finite alphabet  CSPs from other CSPs, we will refer to them as \emph{ordinary CSPs}.

\section{Technical Contributions and  Challenges}
Most closely related to our paper, both at a conceptual and technical level is the paper by~\citet*{GHMRC11} who showed that every ordering CSP is approximation resistant assuming the Unique Games Conjecture (see also~\cite{guruswami2008beating,charikar2009every}). Our main technical contribution is a hardness preserving reduction from ordering CSPs to phylogenetic CSPs. At a high-level, we must deal with three main challenges:

\medskip

\noindent\textbf{Trees Have Topology:} In Phylogenetic CSPs,  whether a phylogenetic constraint is satisfied or not crucially depends on the topology of the tree. Contrast this with what happens in ordering CSPs, where simply knowing the position of an item in the permutation determines if the constraint is satisfied.  For trees, the notion of ``position'' is more complicated and how we split the $n$ items at internal tree nodes is important (see also the discussion on random assignments).

\medskip
\noindent\textbf{Many Types of Trees:} Theorem~\ref{th:main} provides hardness for large collections of problems  studied e.g., in logic, algebra and computational biology, where trees may be ordered (left and right children are distinguishable). Contrast this with ordering CSPs where such considerations are irrelevant.

\medskip
\noindent\textbf{Biased Random Assignment:} Perhaps counterintuitively, even the definition of a ``random tree'' for Phylogenetic CSPs requires some attention. Simply outputting a uniformly random tree on $n$ leaves can result in very poor solutions. Instead, we define a natural ``biased'' version of a random assignment
that generalizes prior methods. We show that it achieves the best possible approximation, under UGC. Contrast this with ordering CSPs, where we simply output a random permutation of $n$ items and this is optimal.

\medskip
In this paper, we present a reduction from ordering CSPs. Let us stress that a na\"ive reduction from ordering CSPs to phylogenetic CSPs does not give the desired hardness results. For example, the Triplet Consistency predicate $uv|w$ can be satisfied when the vertices are ordered as $(u,v,w)$, $(v, u,w)$, $(w, u,v)$, and $(w, v, u)$. So, the best hardness for Triplet Consistency we could hope for if we used the na\"ive reduction would be $4/3! = 2/3$.

Our main technical and conceptual contributions are as follows:
\begin{itemize}
\item We define a new class of biased random assignment algorithms for phylogenetic CSPs with one and many payoff functions
and 
prove matching hardness of approximation results. 
\item We show that the ``gap instance'' from the paper by~\citet{GHMRC11} can be adapted to serve as a ``gadget'' in the reduction from ordering to phylogenetic CSPs. A priori, it is not clear that this gap instance can be used for phylogenetic CSPs because phylogenetic CSPs are quite  different from ordering CSPs. We also modify the hardness reduction by~\citet{GHMRC11} to make it compatible with our own reduction. Their original reduction ``erases the tree structure'' of our instances because it cyclically shifts positions of variables. This preserves the relative order of most $k$-tuples  of variables on the line but not in the tree.
\item We provide a new definition of \emph{coarse} solutions for phylogenetic CSPs. This definition substantially differs from the definition of \emph{coarse} solutions for ordering CSPs (\cite{GHMRC11,
charikar2007advantage}). The most important difference is that we need to assign colors to different \emph{buckets} of vertices. Without this new ingredient, it is not possible to show that every solution to the phylogenetic CSP (with payoff function $f$) can be transformed to a better \emph{coarse} solution (for an altered payoff function $f^+$). 
\item Finally, we extend our results to phylogenetic CSPs with more than one payoff function. The best algorithm for such CSPs first finds the best possible biased assignment and then uses it to obtain a random solution.
\end{itemize}

\section{Preliminaries} \label{sec:prelim}
\noindent\textbf{Trees.}
For ease of exposition, this discussion is focused only on \emph{ordered, rooted, binary} trees. Our results also hold for \emph{unordered} and \emph{unrooted} trees since phylogenetic CSPs on rooted unordered and unrooted unordered trees are special cases of phylogenetic CSPs on ordered trees. Let $T=(V,E)$ be a rooted tree with root $r$. A tree $T$ is called ordered if the child nodes of every internal node $u$ are ordered from left to right. In an ordered tree, all leaves are also ordered from left to right as in a planar drawing of that tree.
In Section~\ref{sec:higher-arity-tree}, we discuss an extension of our results for higher arity trees. Let us note that we will  use auxiliary higher arity trees in the proof of our hardness result even for binary trees. From now on, we simply use the word ``tree'' to refer to ordered, rooted trees.

For $u, v \in V$, we say that $u$ lies below $v$ if the path from $u$ to the root $r$ passes through $v$; in this case, we may also say that $v$ lies above $u$. The Lowest Common Ancestor (LCA) of a set of vertices $S\subseteq V$ is the node $u$ that lies above all vertices in $S$ and has the largest distance from $r$ (the LCA node is uniquely determined by the set $S$). 




\paragraph{Tree Homeomorphism.}
We now define a \emph{phylogenetic} payoff function. Given a tree $T$ and $k$ distinct leaves $u_1,\dots,u_k$ of $T$, a \emph{phylogenetic} payoff function $f$ returns a value (payoff) in $[0,1]$. Loosely speaking, this payoff can only depend on the relative positions of leaves $u_1,\dots,u_k$ and their least common ancestors in the tree. Below, we formalize the definition using the notion of homeomorphism for labeled ordered rooted trees. Then, we examine two ways of defining \emph{phylogenetic payoff functions} using \emph{pattern tables} and \emph{formulas with bracket predicates}. Pattern tables correspond to truth tables of ordinary CSPs, and formulas with bracket predicates correspond to formulas with \emph{and}, \emph{or}, \emph{not} predicates for ordinary CSPs. 

\medskip

Consider a graph $G$ and vertex $u$ in $G$ of degree $2$. Let $v_1$ and $v_2$ be the neighbours of $u$. We call 
the following operation \emph{smoothing $u$ out}: Remove vertex $u$ along with edges $(u,v_1)$ and $(u,v_2)$ from $G$ and then add edge $(v_1,v_2)$ to $G$.

\begin{definition}
Consider two ordered  rooted trees $A$ and $B$. Let $u_1,\dots,u_k$ be distinct leaves in $A$ and $v_1,\dots,v_k$ be distinct leaves in $B$. 

\medskip

\noindent I. We say that $A$ with  labeled leaves $u_1,\dots,u_k$ and $B$  with labeled leaves $v_1,\dots,v_k$ are 
isomorphic if there exists an isomorphism of ordered trees $A$ and $B$ that maps every $u_i$ to $v_i$. Note that an isomorphism $g$ of ordered trees must preserve the order of vertices. That is, if $u$ is to the left of $v$, then $g(u)$ must be to the left of $g(v)$. Also, $A$'s root must be mapped to $B$'s root.

\medskip

\noindent II. We say that $A$ with  labeled leaves $u_1,\dots,u_k$ and $B$ with with labeled leaves $v_1,\dots,v_k$ are homeomorphic if $A$ and $B$ can be transformed to isomorphic trees $A'$ and $B'$ using the following three operations: (1) removing every non-labeled leaf in $A$ or $B$ (i.e., any leaf other than $u_1,\dots,u_k$ in $A$; and any leaf other than $v_1,\dots,v_k$ in $B$); (2) smoothing out vertices of degree 2 in $A$ or $B$ (see above for the definition); (3) removing the root if its degree is $1$ and making its only child the new root.
\end{definition}

In the definition of homeomorphic trees, we can assume that $A'$ and $B'$ are \emph{irreducible trees} i.e., trees that cannot be further minimized using operations (1), (2), and (3). Note that each tree with $k$ labeled leaves has a unique irreducible tree because operations (1), (2), and (3) commute. Consequently, the homeomorphic relation between labeled trees is transitive.

We now can formally define \emph{phylogenetic} payoff functions.
\medskip

\begin{definition}
A function $f$ that takes as input a tree $T$ and $k$ leaves $x_1,\dots,x_k$, and returns a value in the range $[0,1]$, is a \emph{phylogenetic payoff function} if it satisfies the following condition: for any two trees $A$ and $B$, and any leaves $u_1,\dots,u_k$ in $A$ and $v_1,\dots, v_k$ in $B$, if $A$ with labels $u_1,\dots,u_k$
and $B$ with labels $v_1,\dots, v_k$ are homeomorphic, then
$f(A, u_1,\dots,u_k) = f(B, v_1,\dots,v_k)$. The value returned by $f$ is called a payoff.
\end{definition}

To simplify notation, we will omit the tree and write $f(x_1,\dots,x_k)$ instead of $f(T, x_1,\dots,x_k)$ when it is clear that $x_1,\dots,x_k$ are leaves of $T$. 
In this paper, we will only deal with payoff functions whose maximum payoff equals $1$. We call such functions satisfiable.
Before, we proceed to the definition of \emph{phylogenetic CSPs}, we discuss how to define phylogenetic functions using  \emph{tree patterns} and \emph{formulas with bracket predicates}.

\medskip

\noindent\textbf{Tree Patterns.} Intuitively, a \textit{pattern} (or motif) is a small graph that we want to find in a larger graph. Here, we are interested in \textit{tree patterns}.

\begin{definition}
A tree pattern $P$ is a tree with $k$ leaves that are labeled by $k$ variable names $x_1,\dots,x_k$.
\end{definition}
We refer the reader to Section~\ref{sec:figures} for examples of different tree patterns. 

\begin{definition}
Consider a  tree $T$ and $k$ 
leaves $u_1,\dots,u_k$ in $T$. We say that leaves $u_1,\dots,u_k$ match pattern $P(x_1,\dots,x_k)$ in $T$ if 
tree $T$ with labeled leaves $u_1,\dots,u_k$ and $P$ with leaves $x_1,\dots, x_k$ are homeomorphic.
\end{definition}

Every (ordinary) Boolean predicate or function can be specified using a truth table. We now define an analog of a truth table for phylogenetic trees. A pattern table for $f$ is a list of distinct (non-homeomorphic) patterns with $k$ variables $x_1,\dots,x_k$ and payoffs in $[0,1]$
assigned to the patterns. 
The value of phylogenetic payoff function $f$ defined by a pattern table on leaves $u_1,\dots u_k$ equals to the payoff assigned to the pattern $P$ if $u_1,\dots,u_k$ matches $P(x_1,\dots,x_k)$ for some pattern $P$ in the table; and $0$ if $u_1,\dots,u_k$ do not match any pattern in the table. In Figure~\ref{fig:patterns-for-triplets} in Appendix, we show patterns in the pattern table for the Triplet Consistency problem; each pattern in Figure~\ref{fig:patterns-for-triplets} is assigned a payoff of $1$.

\medskip

\noindent\textbf{Bracket Predicates $[a,b<c]$.} We can  specify patterns using  ``bracket predicates.'' Consider three leaves $a,b,c$ of a tree $T$. We say that $[a<b]$ if $a$ appears to the left of $b$ in $T$. We write $[a,b<c]$ if vertices $a$ and $b$ lie in the left subtree of   $\lca(a,b,c)$ and $c$ lies in the right subtree of 
$\lca(a,b,c)$. Similarly, we write
$[a < b,c]$ if $a$ lies in the left subtree of $\lca(a,b,c)$ and $b,c$ lie in the right subtree.
It is not hard to see that every pattern can be expressed as a conjunction of bracket predicates. We show how to represent patterns for the Triplet Consistency payoff function as a conjunction of bracket predicates in Figure~\ref{fig:patterns-for-triplets}. We prove the following Lemma~\ref{lem:bracket-predicates} in Section~\ref{sec:cl:bracket-predicates}.

\begin{lemma}\label{lem:bracket-predicates}
Every pattern can be expressed as a conjunction of bracket predicates.
\end{lemma}
\noindent In Section~\ref{sec:cl:bracket-predicates}, we prove that every phylogenetic payoff function can be defined using a pattern table.

\medskip

\noindent\textbf{Phylogenetic CSPs.} 
A \emph{phylogenetic CSP} problem $\Gamma$ is defined by one or several phylogenetic payoff functions $f_1,\dots, f_A$ of arity $k$. 
An instance $\calI$ of $\Gamma$ consists of a set of variables 
$V$ and sets of $k$-hyperedges $C_{f_i}$ -- one set for each predicate $f_i$ in $\Gamma$. Thus,
$\calI = (V, 
C_{f_1},\dots, C_{f_A})$.
A hyperedge 
$(u_1,\dots,u_k)$ in $C_{f_i}$ represents a constraint 
$f_i$ on variables $u_1,\dots,u_k$.
The weight of a hyperedge $(u_1,\dots,u_k)$ in $C_{f_i}$ 
is denoted by $w_{f_i}(u_1,\dots,u_k)$.

In this paper, we will focus on phylogenetic CSPs with one payoff function. However, all results we prove in this paper also hold for phylogenetic CSPs with multiple payoff functions. We will discuss such CSPs in Section~\ref{sec:multipl-payoff}. When we refer to phylogenetic CSPs with one payoff function, we will omit the index $f_i$ and denote the set of payoff functions by $C$ and weights of constraints by $\weight(u_1,\dots,u_k)$. 

A solution for an instance $\calI$ of phylogenetic CSP $\Gamma$ is an assignment of variables $V$ to leaves of a binary ordered tree $T$ (the tree $T$ is also a part of the solution). We denote the set of all solutions by $\Phi(\calI)$. The value of a solution $\varphi\in \Phi(\calI)$, which we denote by $\val(\varphi,\calI)$, equals the expected value of payoff functions on a random hyperedge in $\calI$:
$$\val(\varphi,\calI) 
= \frac{1}{\weight(\calI)}\sum_{\substack{i\in\{1,\dots,A\}\\(x_1,\dots,x_k)\in C_{f_i}}} \weight(u_1,\dots,u_k) \cdot f(\varphi(u_1),\dots,\varphi(u_k)),$$
where $\weight(\calI)$ is total total weight of all hyperedges (constraints) in $\calI$:
$$\weight(\calI) = \sum_{i}\sum_{(u_1,\dots,u_k)\in C_{f_i}} \weight(u_1,\dots,u_k).$$
We denote the maximum value of a solution $\varphi \in \Phi(\calI)$ by $\opt(\calI)$. We denote the average value of all payoff functions in $\calI$ by 
$\Avg_{(u_1,\dots,u_k)\in \calI} f(\varphi(u_1),\dots,\varphi(u_k))$.

\medskip

\section{Biased Random Assignment and 
Approximation Resistance}\label{sec:random-assignment-def}

In this section, we explore definitions of \emph{randomized assignment}, \emph{biased randomized assignment}, and \emph{approximation resistance}. Recall that a randomized assignment algorithm for ordinary CSPs assigns a random value from the alphabet $\Sigma$ to each variable. Similarly, a random assignment algorithm for ordering CSPs permutes all variables in a random order. An analogue of these algorithms for phylogenetic CSPs  randomly partitions all variables of an instance $\calI$ into two groups and then assigns the first group to the left subtree and the second group to the right subtree. It recursively partitions variables in the left and right subtrees till each node contains at most one variable. This algorithm works well for {\sc{MaxTriplets}} and {\sc{MaxQuartets}}. For these problems, it gives a $\nicefrac{1}{3}$ approximation. However, it drastically fails for some other phylogenetic CSPs. 

Consider the following problem, which we call ``split one node to the right'' (see Figure~\ref{fig:split-one-to-left-right}). The payoff function $s(x_1,\dots,x_k)$ returns $1$ if at every node when three or more variables split, only one of those variables goes to the right subtree and the other variables go to the left subtree. The abovementioned randomized assignment algorithm satisfies a predicate 
$s(x_1,\dots,x_k)$ with
probability exponentially small in $k$. However, a biased randomized assignment algorithm that places every vertex to the left subtree with probability $1-\delta$ and to the right subtree with probability $\delta$ satisfies predicate $s$ with probability close to $1$ if $\delta$ is sufficiently small. In Section~\ref{sec:examples}, we consider a more interesting example. In that example, a randomized assignment algorithm should split vertices into two groups with probabilities that change from one recursive call to another.

The discussion above leads us to the following definition of a biased random assignment algorithm (for biased random assignments for ordinary CSPs, see~\cite*{GL2014-biased}). A biased random assignment algorithm is specified by an absolutely continuous probability measure $\rho$ on the interval $[0,1]$. We remind the reader that $\rho$ is absolutely continuous if there exists a measurable function $h$ such that 
$\rho(S) = \int_S h(t) dt$ for every measurable subset $S$ of $[0,1]$.
The measure of the interval $[0,1]$ equals $1$ because $\rho$ is a probability measure. We assign every node of the infinite complete binary tree a subinterval of $[0,1]$. The root of the tree is assigned $[0,1]$. Its left child is assigned $[0,1/2]$ and right child is assigned $[1/2, 1]$. This assignment defines weights of all nodes -- the weight $\rho(u)$ equals to the measure of the interval corresponding to $u$. 

We now assume that the algorithm is given oracle access to $\rho$. The  algorithm recursively partitions variables in $\calI$. Initially, it assigns all variables to the root of the binary tree. At every step, the algorithm picks a node $u$ of the binary tree that contains more than one variable,  creates two child nodes 
$u_{left}$ and $u_{right}$,
and randomly splits all variables in $u$ between $u_{left}$ and $u_{right}$. Namely, it assigns each $x$ in $u$ to $u_{left}$ and $u_{right}$ with probabilities $\rho(u_{left})/\rho(u)$ and $\rho(u_{right})/\rho(u)$, respectively. 

Let $\alpha_{\rho}(f)$ be the approximation factor of the biased random assignment algorithm with measure $\rho$ for payoff function $f$ and $\alpha_{opt}(f)$ be the best approximation factor of a biased randomized assignment algorithm for payoff function $f$:
$$\alpha_{opt}(f) = 
\sup_{\rho}\alpha_{\rho}(f).$$
As we mentioned earlier, we now consider phylogenetic CSPs with one payoff function~$f$. If phylogenetic CSP $\Gamma$ has several payoff functions, then it should first randomly choose the appropriate measure $\rho$. We discuss such CSPs in Section~\ref{sec:multipl-payoff}.
If a phylogenetic CSP $\Gamma$ has only one payoff function, we will write $\alpha_{opt}(\Gamma) = \alpha_{\opt}(f)$. We call 
$\alpha_{opt}(\Gamma)$ the random assignment approximation factor for $\Gamma$.

We note that every measurable function $h$ can be approximated by a piecewise constant function $h'$. Function $h'$ is a constant on each interval $S_1,\dots, S_q$ that partition $[0,1]$ into $q$ parts. Moreover, we can assume that the enpoints of intervals $S_i$ are binary rational numbers. This lets us define 
$\rho'$-biased algorithms in the following equivalent way. A $\rho'$-biased algorithm is defined by a constant-size tree $T$ and a probability distribution $\rho'$ on the leaves of $T$. The algorithm first assigns every variable to one of the leaves of $T$ using the distribution $\rho'$. Then, it recursively partitions variables splitting them 50\%/50\% at every step. The running time of this algorithm is linear in $n$ (the number of variables). We have the following theorem.

\begin{theorem}
For every phylogenetic CSP $\Gamma$ and every positive $\varepsilon$, there exists a linear-time biased randomized rounding algorithm that has an approximation factor of $\alpha_{opt}(\Gamma) - \varepsilon$.
\end{theorem}

We discuss the case of CSPs with more than one payoff section in Section~\ref{sec:multipl-payoff}. Finally, we define approximation resistance for phylogenetic CSPs.

\begin{definition} A phylogenetic CSP $\Gamma$ is approximation resistant if for every positive $\varepsilon$, it is NP-hard to distinguish between instances of $\Gamma$ that (a) have a solution of value greater than $1-\varepsilon$ and (b) do not have a solution of value greater than 
$\alpha_{opt}(\Gamma)+\varepsilon$.
\end{definition}

We show that all phylogenetic CSPs are approximation resistant. In particular, this means that, unless $P=NP$ (assuming UGC), for every phylogenetic CSP $\Gamma$, there is no approximation algorithm with a constant approximation factor better than $\alpha_{opt}(\Gamma)+\varepsilon$.

\section{Proof Overview}\label{sec:proof-overview}
In this section, we give an overview of our hardness result for phylogenetic CSPs. We first show that the Triplets Consistency problem ({\sc{MaxTriplets}}) is approximation resistant by providing a reduction from a specially crafted ordering CSP problem to the Triplets Consistency problem. This reduction works for Triplets Consistency and some other phylogenetic problems, however, fails in the general case. We then show how to modify the construction by \cite*{GHMRC11} to make our reduction work for all phylogenetic CSPs.

Consider a phylogenetic CSP $\Gamma_{phy}$. In this section, we assume that this phylogenetic CSP has only one payoff function $f$ of arity $k$. We will discuss the case when $\Gamma_{phy}$ has several payoff functions in Section~\ref{sec:multipl-payoff}. Let $\calI$ be an instance of $\Gamma_{phy}$. Denote the value of the optimal solution for $\calI$ by  $\opt(\calI)$. 
Observe that every solution $\varphi$ to $\calI$ defines an ordering on the variables of the instance $\calI$. 
In this ordering, the variables are arranged from left to right according to their position in the embedding $\varphi$
in the binary tree. We denote the ordering of the variables in $\calI$ by $\order(\varphi)$. 
Let $\Phi_{\pi}(\calI)$ be the set of all solutions for instance $\calI$ in which the order of variables is $\pi$ (i.e., $\order(\varphi) = \pi$); and let  $\opt(\calI\mid\pi)$ be the value of the best solution $\varphi$ in $\Phi_{\pi}(\calI)$:

\begin{equation}\label{eq:opt-I-given-pi}
\opt(\calI\mid\pi)
= \max_{\substack{\varphi\in\Phi_{\pi}(\calI)}}
\val(\varphi,\calI).
\end{equation}

\medskip

\noindent\textbf{Gap instance.} Following~\citet{GHMRC11}, we use a \emph{gap instance} $\calI^{f}_{gap}$ in our reduction. The variables of this instance are leaves of an ordered perfect  $k$-ary tree of depth $d$ (in this tree all internal nodes have $k$ children; and the depth of all leaves is $d$). Each constraint in the instance is a payoff function $f$ on a subset of $k$ leaves/variables. To define the instance, we introduce a random map $L_{k,m}:k\to V$, where $V$ is the set of leaves and $m=|V|$. Random map $L_{k,m}$ works as follows: it picks a random $i$ from set $\{0,1, \dots, d-1\}$ and then a random (internal) node $u$ at level $i$ of the tree $T$. Let $u_1,\dots, u_k$ be the child nodes of $u$ arranged from left to right. In the subtrees $T_{u_1},\dots, T_{u_k}$ rooted at vertices $u_1,\dots, u_k$, we independently pick random leaves $l_1,\dots, l_k$ and map each $i$ to $l_i$. Now, for every $k$ vertices $l_1,\dots, l_k$, we add hyperedge 
$(l_1,\dots,l_k)$ to the set of constraints $C$. The weight of 
$(l_1,\dots,l_k)$
equals
$\Pr\{L_{k,m}(1) = l_1,\dots, L_{k,m}(k)=l_k\}$. Note that we use exactly the same gap instance as~\cite{GHMRC11}. In their instance, the payoff function is an ordering payoff function. We will use this instance with ordering, phylogenetic, and ordinary payoff functions. In fact, we will think of $\calI_{gap}$ as a template for $k$-ary CSP instances (formally, 
$\calI_{gap}=C$ is the set of hyperedges). Then, $\calI^f_{gap} = (f,C)$ is an instantiation of this template with the payoff function $f$.

\medskip

\noindent \citet{GHMRC11} showed that
\begin{enumerate}[I.]
\item $\calI_{gap}^f$ is a completely satisfiable instance of an ordering CSP for every ordering payoff function $f$ with $f(id) = 1$. That is, if $f(id) = 1$, then $\opt(\calI_{gap}^{f}) = 1$. Note that for every satisfiable payoff function $f$, we can rearrange its inputs using some permutation $\sigma$ so that $f\circ\sigma(id) = 1$.
\item The cost of any so-called \emph{coarse} solution to this instance is at most $\alpha + \varepsilon$, where $\alpha$ is the expected value of the random assignment algorithm (which is unique for ordering CSPs unlike phylogenetic CSPs), and $\varepsilon$ tends to $0$ as the depth $d$ (see above) of tree $T$ tends to infinity. 
\end{enumerate}
In our proof, we will also use coarse solutions. However, we postpone the discussion of such solutions till  Section~\ref{sec:coarse-solution}, where we define coarse solutions for phylogenetic CSPs (note that coarse solutions for phylogenetic CSPs and ordering CSPs differ in several important ways). In this paper, we will use the following lemma, which is an analog of property II above.
\begin{lemma}\label{lem:ordinary-CSPs-gap}
Fix natural numbers $k\geq 1$, $q\geq 1$ and positive real number $\varepsilon > 0$.
Then, there exists a natural $m^*$ such that for every template $\calI_{gap}$ with at least $m^*$ leaves from the family defined above, every (ordinary) payoff function $f$ of arity $k$ defined on alphabet $\Sigma$ of size $q$ (i.e., $f$ is a function from $\Sigma^k$ to $[0,1]$), the following claim holds:
\begin{equation}\label{eq:lem:gap-I-for-CSPs}
\opt(\calI_{gap}^f)\leq 
\max_{\rho} \E_{x_i\sim \rho}\big[
f(x_1,\dots,x_k)\big] + \varepsilon,    
\end{equation}
where $\rho$ is a probability distribution on $\Sigma$; and $x_1,\dots,x_k$ are drawn from $\rho$ independently.
\end{lemma}
The lemma is similar to Theorem 11.1 from the paper by~\cite*{GHMRC11}. For completeness, we provide a proof of Lemma~\ref{lem:ordinary-CSPs-gap} in Section~\ref{sec:ordinary-CSPs-gap}. To prove Lemma~\ref{lem:ordinary-CSPs-gap}, we rewrite bound (\ref{eq:lem:gap-I-for-CSPs}) as a bound on the KL-divergence of certain random variables. Then, we prove the new bound on the KL-divergence using the chain rule for conditional entropy. We believe that our approach is somewhat simpler than the original approach used by~\cite{GHMRC11}.

We also show that gap instance $\calI_{gap}^{f}$
is completely satisfiable. We present a proof of the following lemma in Section~\ref{sec:gap-sat}. Note that this statement is not obvious for phylogenetic CSPs and does not follow from the previously known results. 
\begin{lemma}\label{lem:sat-I}
For every  phylogenetic CSP $\Gamma$ with payoff function $f$ and gap instance $\calI_{gap}^f$ of arbitrary size, we have
$\opt(\calI_{gap}^f) = 1$.
\end{lemma}

The main technical tool in our reduction is Theorem~\ref{thm:exists-gap-instance}. This theorem shows that $\opt(\calI_{gap}^f|\pi)
\approx \alpha$ for a uniformly random permutation $\pi$. 
A crucial ingredient in the proof of this theorem is a new definition of coarse solutions and colorings, which we discuss in Section~\ref{sec:coarse-solution}.
\begin{theorem}\label{thm:exists-gap-instance}
Consider a phylogenetic CSP $\Gamma$ with a payoff function $f$. For every positive $\varepsilon > 0$, there exists a sufficiently large gap instance $\calI^f_{gap}$ of phylogenetic CSP $\Gamma$ such that 
$\opt(\calI^f_{gap})=1$; and
for a random permutation $\pi$ of variables of $\calI_{gap}^f$:
$$\E_{\pi}[\opt(\calI^f_{gap} \mid \pi)]  \leq \alpha + \varepsilon.$$
\end{theorem}
We prove this theorem in Section~\ref{sec:fill-gap}. We now discuss how to use this theorem to construct a gap preserving reduction from a certain ordering CSP problem to $\Gamma_{phy}$.
Fix $\varepsilon > 0$. Let  $\calI_{gap}^f$ be a gap instance from Theorem~\ref{thm:exists-gap-instance} and $m$ be the number of variables in $\calI_{gap}^f$.
Define an auxiliary ordering CSP $\Gamma_{ord}$ with a payoff function $o$ of arity $m$. The value of $o(\pi)$ on variables $x_1,\dots,x_m$ equals
$$o(\pi(x_1),\dots,\pi(x_m)) = \opt(\calI^f_{gap}\mid\pi),$$
where the instance $\calI_{gap}^f$ is also defined on $x_1,\dots,x_m$. In other words, the value of payoff function $o$ with variables $x_1,\dots, x_m$ on permutation $\pi$ equals to the best solution $\varphi$ for 
instance $\calI^f_{gap}$ with the same set of variables $x_1,\dots, x_m$ subject to the constraint that the variables in  solution $\varphi$ are ordered according to $\pi$.

\medskip

\noindent\textbf{Reduction.}  Let 
$\Gamma_{phy}$ be the class of phylogenetic CSPs with payoff function $f$, and $\Gamma_{ord}$ be the class of ordering CSPs with payoff function $o$. We now define a reduction $h_{ord\to phy}$ from CSPs $\Gamma_{ord}$ to CPSs $\Gamma_{phy}$.
We take an arbitrary instance $\calI_{ord}$ of $\Gamma_{ord}$ and transform it to an instance $\calI_{phy}$ of $\Gamma_{phy}$ on the same set of variables as $\calI_{ord}$. In instance $\calI_{ord}$, we
replace every constraint $(x_{i_1},\cdots,x_{i_m})$ for payoff function $o$ with a copy of the gap instance $\calI^f_{gap}$ on variables $x_{i_1},\cdots,x_{i_m}$. 
That is, $\calI_{phy}$ is the union of copies of the gap instances (``gadgets'') $\calI_{gap}^f$ -- one ``gadget'' for every constraint $(x_{i_1},\dots,x_{i_m}$ in $\calI_{ord}$. We denote the obtained instance of phylogenetic CSP $\Gamma_{phy}$ by $\calI_{phy}$. We let 
$h_{ord\to phy}(\calI_{ord}) = \calI_{phy}$.

Note that for every solution $\varphi$ to the phylogenetic CSP  $\calI_{phy}$, there is a 
corresponding solution $\pi$ to the ordering CSP $\calI_{ord}$. This solution $\pi$ orders all variables in $\calI_{ord}$ in the same way as they are ordered by solution $\varphi$ in the phylogenetic tree for $\calI_{phy}$ i.e., $\pi=\order(\varphi)$.
The value of each payoff function $o$ on $\pi$ is greater than or equal to the value 
of $\varphi$ on the copy of $\calI_{gap}^f$ created for $o$. This is the case, because $\varphi$ is a possible solution for that copy of $\calI_{gap}^f$ (since the variables in $\varphi$ are ordered according to $\pi$). We have the following claim.

\begin{claim}\label{lem:reduction-decrease-value}
Consider instances $\calI_{ord}$ to $\calI_{phy}$ of \emph{ordering} and phylogenetic CSPs as above. Then, 
$$\opt(\calI_{ord}) \geq \opt(\calI_{phy}).$$
\end{claim}

Unfortunately, we cannot claim that $\opt(\calI_{ord}) = \opt(\calI_{phy})$. It is possible that $\opt(\calI_{ord}) \gg \opt(\calI_{phy})$. This may happen if there exists an ordering of variables $\pi$ such that for every constraint $\mathbf{u} = (u_1,\dots,u_m)$
in $\calI_{ord}$, there exists a good \emph{local} solution
$\varphi_{\mathbf{u}}\in \Phi_{\pi}(\calI)$:
$$
o(
\varphi_{\mathbf{u}}(u_1)
,\dots,
\varphi_{\mathbf{u}}(u_m))
\approx 1.
$$
Note that $\varphi_{\mathbf{u}}$ may depend on the constraint $\mathbf{u}$.
However, there is no good \emph{global} solution $\varphi\in \Phi_{\pi}(\calI)$ that works for all constraints $(u_1,\dots,u_m)$ (on average) in $\calI_{phy}$. That is,
for every $\varphi\in \Phi_{\pi}(\calI)$,
$$
\val(\varphi,\calI_{phy}) = \Avg\limits_{(u_1,\dots,u_m)\in \calI_{phy}}
o(\varphi(u_1),\dots,\varphi(u_m))\ll 1.
$$

\medskip

\noindent\textbf{Hardness of approximation.}
We now discuss how to use  reduction $h_{ord\to phy}$ to show that $\Gamma_{phy}$ is approximation resistant.
The ordering CSP $\Gamma_{ord}$ is approximation resistant as every ordering CSP (\cite*{GHMRC11}). By Theorem~\ref{thm:exists-gap-instance}, the expected value of payoff function $o$ on a random permutation $\pi$ is at most $\alpha+\varepsilon$. Hence, assuming the Unique Games conjecture, it is NP-hard to distinguish between
\begin{enumerate}[A.]
\item instances of $\Gamma_{ord}$ that are at most $(\alpha + \varepsilon) + \varepsilon$ satisfiable; and
\item instances of $\Gamma_{ord}$ that are at least $(1-\varepsilon)$ satisfiable.
\end{enumerate}

To finish the proof, we would like to show that $h_{ord\to phy}$ is a gap preserving reduction. Namely, $h_{ord\to phy}$ maps (a) every instance $\calI_{ord}$ of value at most $\alpha +2\varepsilon$ to an instance $\calI_{phy}$ of value at most $\alpha +2\varepsilon$;
and (b) every instance $\calI_{ord}$ of value at least $1-\varepsilon$ to an instance of $\Gamma_{phy}$
of value $1-O(\varepsilon)$. If $h_{ord\to phy}$ satisfied these properties, we would conclude that, assuming UGC, it is NP-hard to  distinguish between
(A) instances of $\Gamma_{phy}$ that are at most $\alpha + 2\varepsilon$ satisfiable; and (B) instances of $\Gamma_{phy}$ that are at least $1-O(\varepsilon)$ satisfiable.

Property (a) immediately follows from Claim~\ref{lem:reduction-decrease-value} because reduction $h_{ord\to phy}$ does not increase the value of the instance. Unfortunately, property (b) is not satisfied for many payoff functions $f$. Nevertheless, in the next section, we show that property (b) holds for one particular function $f^*$ and, consequently, the phylogenetic CSP with that payoff function $f^*$ is approximation resistant. We will use this result to prove that Triplets Consistency is also approximation resistant.

In Section~\ref{sec:reduction-ug}, we will deal with arbitrary phylogenetic CSPs. Specifically, we will modify the hardness reduction by
\citet{GHMRC11} and obtain a reduction $h_{UG\to ord}$ from Unique Games to $\Gamma_{ord}$ such that the composition of reductions 
$$h_{UG\to phy} =
h_{ord\to phy}\circ h_{UG\to ord}$$
is
gap preserving.

\subsection{Hardness for Triplets Consistency}
In this section, we define a special payoff function $f^*$ of arity $3$ for which the hardness reduction $h_{ord\to phy}$ (described in the previous section) maps almost satisfiable instances of $\Gamma_{ord}$ to  
almost satisfiable instances of $\Gamma_{phy}$. Fix a small $\delta\in(0,1)$. Let $\triplet(u,v,w)$ be the Triplet Consistency payoff function: $\triplet(u,v,w) = 1$, if  $\lca(u,w) = \lca(v,w)$ (in other words, $w$ is separated from $u$ and $v$ before $u$ and $v$ are separated); $\triplet(u,v,w) = 0$, otherwise. Now let 
$f^*(u,v,w) = \triplet(u,v,w)$ if the ordering of variables $u$, $v$, and $w$ in the phylogenetic tree is $u$, $v$, and $w$; $f^*(u,v,w) = (1-\delta)\triplet(u,v,w)$, otherwise. Observe that $f^*(u,v,w)$ is a satisfiable payoff function
i.e., its maximum value is $1$. 
Let $\Gamma_{phy}$ be the phylogenetic CSP with payoff function $f^*$ and 
$\Gamma_{ord}$ be the corresponding ordering CSP.

\begin{lemma}\label{lem:f-star-payoff}
Reduction $h_{ord\to phy}$ maps every $(1-\varepsilon)$-satisfiable instance of $\Gamma_{ord}$ to a $(1-\varepsilon/\delta)$-satisfiable instance of $\Gamma_{phy}$.
\end{lemma}
\begin{proof}
Consider a $(1-\varepsilon)$-satisfiable instance $\calI_{ord}$ 
of ordering CSP $\Gamma_{ord}$ and the corresponding instance 
$\calI_{phy} = h_{ord\to phy}(\calI_{ord})$ of phylogenetic CSP $\Gamma_{phy}$. Let $\pi$ be the optimal solution to $\calI_{ord}$. Consider the ``left'' caterpillar binary tree $T$ with $n$ leaves.
Tree $T$ is a binary tree in which the right child of every internal node is a leaf. We construct $T$ by taking a path of length $n$ and attaching a right child to every node but last (see Figure~\ref{fig:caterpillar-binary-tree} in Appendix). We now define a solution for instance $\calI_{phy}$ that maps every variable of $\calI_{phy}$ to a leaf of $T$.  We number all leaves in the tree from left to right. Then, we map every variable $u$ to the leaf number $\pi(u)$. Thus, the ordering of variables in solution $\varphi$ is $\pi$.

We prove that $\val(\varphi, \calI_{phy})\geq 1-\varepsilon/\delta$. Observe that $f^*(\varphi(u),\varphi(v),\varphi(w)) = 1$ for a triplet $(u,v,w)$ if and only if $\pi(u) < \pi(v) < \pi(w)$ (because $\varphi$ maps all vertices to the leaves of the left caterpillar tree). So, it is sufficient to show that $\pi(u) < \pi(v) < \pi(w)$ for all but at most $1-\varepsilon/\delta$ fraction of all constraints in $\calI_{phy}$. In other words, we need to show
$$
\Avg_{(u,v,w)\in\calI_{phy}} 
\ONE\big(\pi(u)< \pi(v) < \pi(w)\big) \geq 1 - \varepsilon/\delta,
$$
where $\ONE(\pi(u)< \pi(v) < \pi(w))$ is the indicator of the event 
$\pi(u)< \pi(v) < \pi(w)$. Recall that for every ordering constraint $\mathbf{x}=(x_1,\dots,x_m)$ in $\calI_{ord}$, we created a copy $\calI_{\mathbf{x}}$ of the gap instance 
$\calI_{gap}^{f^*}$. Instance
$\calI_{phy}$ is the union of instances 
 $\calI_{\mathbf{x}}$ over all constraints 
$\mathbf{x}$ in $\calI_{ord}$. Thus,
\begin{equation}
\label{eq:opt-ord-using-triplets}
\Avg_{(u,v,w)\in\calI_{phy}} 
\ONE\big(\pi(u)< \pi(v) < \pi(w)\big) 
=
\Avg_{\mathbf{x}\in\calI_{ord}} 
\Avg_{(u,v,w)\in\calI_{\mathbf{x}}} 
\ONE(\pi(u)< \pi(v) < \pi(w)). 
\end{equation}
Consider an ordering constraint $\mathbf{x}=(x_1,\dots,x_m)$ in $\calI_{ord}$. The value of the ordering payoff function $o$ on $\mathbf{x}$ equals (by the definition of $o$):
$$o(\pi(x_1),\dots,\pi(x_m)) = \opt(\calI^f_{gap}\mid\pi) = 
\max_{\varphi_{\mathbf{x}}\in \Phi_{\pi}(\calI_{phy})}
\Avg_{(u,v,w)\in\calI_{\mathbf{x}}} 
f^*(
\varphi_{\mathbf{x}}(u),
\varphi_{\mathbf{x}}(v),
\varphi_{\mathbf{x}}(w)).
$$
Observe that 
$$
f^*(
\varphi_{\mathbf{x}}(u),
\varphi_{\mathbf{x}}(v),
\varphi_{\mathbf{x}}(w)) \leq 
(1-\delta) + \delta\cdot \ONE\big(\pi(u)< \pi(v) < \pi(w)\big)
.$$
This is because every $\varphi_{\mathbf{x}}$ in $\Phi_{\mathbf{x}}(\calI_{phy})$
must order $u$, $v$, $w$ according to permutation $\pi$. So, if 
$\ONE(\pi(u)< \pi(v) < \pi(w)) = 0$, then 
$f^*(
\varphi_{\mathbf{x}}(u),
\varphi_{\mathbf{x}}(v),
\varphi_{\mathbf{x}}(w)) \leq 1 - \delta$. Therefore,
$$
o(\pi(x_1),\dots,\pi(x_m))\leq 
(1-\delta) + \delta \cdot
\Avg_{(u,v,w)\in\calI_{\mathbf{x}}}\ONE(\pi(u)< \pi(v) < \pi(w)).$$
Since $\pi$ satisfies at least $(1-\varepsilon)$ fraction of all constraints in $\calI_{ord}$, we have
$$
(1-\delta) + \delta\cdot
\Avg_{\mathbf{x}\in\calI_{ord}} 
\Avg_{(u,v,w)\in\calI_{\mathbf{x}}} 
\ONE\big(\pi(u)< \pi(v) < \pi(w)\big)\geq 1 - \varepsilon.$$
This inequality implies~(\ref{eq:opt-ord-using-triplets}). This concludes the proof of Lemma~\ref{lem:f-star-payoff}.
\end{proof}

By Lemma~\ref{lem:reduction-decrease-value} and Lemma~\ref{lem:f-star-payoff}, reduction $h_{ord\to phy}$ maps 
(a) instances of $\Gamma_{ord}$ with value at most $\alpha' + 2\varepsilon$ to instances of $\Gamma_{phy}$ also with value at most $\alpha' + 2\varepsilon$; and
(b) almost satisfiable instances of $\Gamma_{ord}$ to almost satisfiable instances of $\Gamma_{phy}$, where $\alpha'$ is the value of the best biased random assignment for $f^*$. Therefore, phylogenetic CSP $\Gamma_{phy}$ with payoff function $f^*$ is approximation resistant.

We now show that the Triplets Consistency problem is also approximation resistant. First, observe that 
$\triplet(u,v,w) \leq f^*(u,v,w)$ for all variables $u,v,w$. Hence, $\alpha' \leq \alpha = 1/3$, where $\alpha$ is the value of the best biased random assignment for the Triplets Consistency problem. Then, note that a $(1-\varepsilon)$-satisfiable instance of the problem with payoff function $f^*$ is also a
$(1-\varepsilon)$-satisfiable instance of the problem with payoff function $\triplet$ (simply because 
$\triplet(u,v,w)\geq f^*(u,v,w)$). Finally, every instance with value at most $\alpha + \varepsilon$ with payoff function $f^*$ has value 
at most $(\alpha + \varepsilon)/(1-\delta)$ with payoff function $\triplet$. This implies that the Triplets Consistency problem is approximation resistant. 
\vspace{-0.4cm}
\paragraph{Hardness for Quartets Consistency.} Using our hardness results for triplets, a simple reduction then proves that {\sc{MaxQuartets}} is also approximation resistant (see Claim~\ref{cl:triplets-to-quartets} in Appendix~\ref{sec:triplets-to-quartets}):
\begin{corollary}
Unrooted Quartet Reconstruction ({\sc{MaxQuartets}}) is approximation resistant, so it is UGC-hard to beat the (trivial) random assignment algorithm that achieves a $\tfrac13$-approximation.
\end{corollary}
\section{Filling the Gaps}\label{sec:fill-gap}
In this section, we build machinery to prove Theorem~\ref{thm:exists-gap-instance}. First, we show that every gap instance $\calI_{gap}^f$ of a phylogenetic CSP with payoff function $f$ is completely satisfiable. Then, we introduce coarse solutions for phylogenetic CSPs and prove important results about such solutions. In the end of this section, we put all parts together and prove
Theorem~\ref{thm:exists-gap-instance}.

\subsection{Gap Instance is Completely Satisfiable}\label{sec:gap-sat}

Consider a phylogenetic CSP with a satisfiable payoff function $f$ of arity $k$. Let $P$ be the pattern of $f$ with a payoff of $1$. Pattern $P$ is a tree with $k$ leaves $l_1,\dots,l_k$ such that $f(l_1,\dots, l_k) = 1$. By permuting the arguments of the payoff function $f$, we may assume that the leaves $l_1,\dots, l_k$ are ordered from left to right in $P$.
We now show that $\calI_{gap}^f$ is a satisfiable instance of $\Gamma_{phy}$. We will need the following definition. 

\begin{definition}\label{def:cousins}
Consider $k$ leaves $l_1,\dots, l_k$ of a full tree $T$ of arity $k$. Let $u$ be their least common ancestor and $u_1,\dots, u_k$ be the child nodes of $u$ ordered from left to right. We say that $l_1,\dots, l_k$ are \emph{cousins} if each $l_i$ is a leaf in the subtree $T_{u_i}$ rooted at $u_i$. We also define a predicate $\cousins$:
$\cousins(l_1,\dots,l_k) = 1$, if 
$l_1,\dots,l_k$ are cousins; and 
$\cousins(l_1,\dots,l_k) = 0$, otherwise.
\end{definition}

\begin{lemma}\label{lem:cousins-sat}
Let $f$ be a satisfiable phylogenetic payoff function and $P$ be a pattern as above. Consider an instance $\calI=(V,C)$ of phylogenetic CSP with a payoff function $f$ and a mapping $\psi$ of variables $V$ of $\calI$ to a $k$-ary tree $T$. Then, there exists a binary tree $T'$ with the same set of leaves as $T$ such that the following statement holds: If $x_1,\dots, x_k$ are mapped to cousins in $T$ by $\psi$, then $f(\psi(x_1),\dots,\psi(x_k)) = 1$ in $T'$.
\end{lemma}
Observe that in the gap instance 
$\calI_{gap}^f$ all payoff functions are defined on leaves that are cousins. Hence, the conditions of Lemma~\ref{lem:cousins-sat} are satisfied for the identity map $\psi$. Therefore, we have the following  immediate corollary.

\begin{corollary}\label{cor:gap-sat}
Every gap instance $\calI_{gap}^f$ 
with phylogenetic payoff function $f$ as above is completely satisfiable.
\end{corollary}
\begin{proof}[Proof of Lemma~\ref{lem:cousins-sat}]
Let $r$ be the root of pattern $P$ and $l_1,\dots, l_k$ be its leaves. We build a binary tree $T'$ by replacing every node and its children in $T$ with the pattern $P$. See Figure~\ref{fig:satisfiable-CSP}. Formally, we define $T'$ as follows. For every internal vertex $u$ of $T$, we create a copy of pattern $P$. Denote it by $P^u$. We identify every vertex $u$ of $T$ with the root of $P^u$. Also,
we identify the $i$-th leaf of $P^u$ with the $i$-th child of $u$. Now consider a payoff function $f$ on variables $x_1,\dots, x_k$. Suppose that $\psi(x_1),\dots, \psi(x_k)$ are cousins in tree $T$. Let $u$ be their least common ancestor and $u_1,\dots,u_k$ be $u$'s child nodes. Then,
each $\psi(x_i)$ lies in the subtree rooted in $u_i$. Let us now examine where $\psi(x_1),\dots, \psi(x_k)$ are located in the new tree $T'$. Each $\psi(x_i)$ also lies in the subtree rooted at $u_i$. However, in $T'$, 
$u_1,\dots u_k$ are not child nodes of $u$ but rather leaves of a copy of the pattern $P$. 
Thus, $\psi(x_1),\dots, \psi(x_k)$ match pattern $P$ in $T'$. Consequently, 
$f(\psi(x_1),\dots,\psi(x_k)) = 1$ for solution $\psi$ on phylogenetic tree $T'$.
\end{proof}

\subsection{Coarse Solutions, Labelling, and Coloring}\label{sec:coarse-solution}
In this section, we define coarse solutions for phylogenetic CSPs and discuss how to measure the value of such solutions. A coarse solution embeds the set of variables $V$ into leaves of a binary tree $T$. Unlike a true solution for a phylogenetic CSP, in a coarse solution, many variables can and, in most cases, will be mapped to the same leaf. A coarse solution also assigns a color to every leaf of $T$. We denote the leaf assigned to variable $x$ by $\xi(x)$ and color assigned to the leaf by $\clr(\xi(x))$.
We say that a coarse solution $\xi$ is in class $\Xi_{\varepsilon, q,\pi}(\calI)$ (where $\varepsilon\in \bbR^+$,
$q\in \bbN$, $\pi$ is an ordering of $V$)
if it satisfies the following properties:%
\footnote{These conditions are slightly more complex for phylogenetic CSPs on non-binary trees.
}
\begin{enumerate}
\item (coarse) tree $T$ has at most $q$ leaves;
\item at most $\varepsilon |V|$ distinct variables have the same color; and
\item moreover, variables mapped to a leaf $l$ are consecutive variables in ordering $\pi$.
\end{enumerate}

Note that this definition differs a lot from the definition of a coarse solution for ordering CSPs. In particular, coarse solutions for ordering CSPs do not assign colors to variables.

We now define two value functions for a coarse solution $\xi$. Consider an instance  $\calI=(V,C)$ of phylogenetic CSP with payoff function $f$ and an arbitrary constraint $(x_1,\dots,x_k)\in C$. If all variables $x_1,\dots, x_k$ have distinct colors in the coarse solution i.e., $\clr(\xi(x_i))\neq \clr(\xi(x_j))$ for all $i,j$, then we let 
$$
f^-(\xi(x_1),\dots,\xi(x_k)) = 
f^+(\xi(x_1),\dots,\xi(x_k)) = 
f(\xi(x_1),\cdots,\xi(x_k)),
$$
here $f(\xi(x_1),\cdots,\xi(x_k))$ is well defined because all leaves $\xi(x_1),\dots,\xi(x_k)$
are distinct. If, however, two variables have the same color (i.e., $\clr(\xi(x_i))= \clr(\xi(x_j))$ for some $i,j$), then we let
$f^-(\xi(x_1),\dots,\xi(x_k))
= 0$ and
$f^+(\xi(x_1),\dots,\xi(x_k)) = 1$. We then define
\begin{align}   
\label{eq:def:val-pm-1}
\val^-(\xi,\calI) &= 
\Avg_{(x_1,\dots,x_k)\in\calI}f^-(\xi(x_1),\dots,\xi(x_k))
\;\;\;\text{and}\;\;\;\\
\val^+(\xi,\calI) &= 
\label{eq:def:val-pm-2}
\Avg_{(x_1,\dots,x_k)\in\calI}f^+(\xi(x_1),\dots,\xi(x_k)).
\end{align}
In both expressions above, we are averaging over all 
constraints $(x_1,\dots,x_k)$ in instance $\calI$.

We will use coarse instances and value functions $\val^+$, $\val^-$ to prove Theorem~\ref{thm:exists-gap-instance}. Our plan is as follows. We first show that for every (true) solution $\varphi\in \Phi_{\pi}$ for instance $\calI$, there exists a coarse solution 
$\xi\in \Xi_{\varepsilon, q, \pi}$ with 
$\val^+(\xi, \calI) \geq \val(\varphi,\calI)$ (see Lemma~\ref{lem:good-cs}). We then argue that for a random ordering $\pi$ and $\xi\in \Xi_{\varepsilon, q, \pi}$, we have $\val^+(\xi, \calI) - \val^-(\xi, \calI) \leq \varepsilon$ with high probability (see Lemma~\ref{lem:diff-val-pm} for the precise statement). Loosely speaking, this is the case because for a random ordering $\pi$, the expected fraction of constraints $(x_1,\dots,x_k)$ with at least two variables having the same color is very small; but $\val^-(\xi,f)$ and $\val^+(\xi,f)$ differ only on such constraints. Finally, we use Lemma~\ref{lem:no-good-cs-for-gap} to show that $\val^-(\xi, \calI_{gap}^f)\leq \alpha+\varepsilon$. The above chain of inequalities implies that 
$$
\opt(\calI_{gap}^f \mid \pi) =
\max_{\varphi\in \Phi_{\pi}} \val(\varphi, \calI_{gap}^f) \leq 
\max_{\xi\in\Xi_{\varepsilon,q,\pi}}
\val^+(\xi, \calI_{gap}^f)
\leq
\max_{\xi\in\Xi_{\varepsilon,q,\pi}}
\val^-(\xi, \calI_{gap}^f)
+\varepsilon \leq \alpha + 2\varepsilon
$$
with high probability if  $\pi$ is a random ordering of variables in 
$\calI_{gap}^f$.

\subsection{How to Transform True Solution to Better Coarse Solution?}\label{sec:good-cs}
We now show that for every true solution $\varphi$ for $\calI$, there exists a  coarse solution with  $\val^+(\xi,\calI) \geq \val(\varphi,\calI)$. In this coarse solution $\xi$ the variables are ordered in the same way as in $\varphi$.

\begin{lemma}\label{lem:good-cs}
Consider an instance $\calI=(V,C)$ of a phylogenetic CSP $\Gamma_{phy}$. Let $\varepsilon > 1/|V|$. For every permutation $\pi$ and every solution $\varphi\in\Phi_{\pi}(\calI)$ for $\calI$, there exists a coarse solution $\xi \in \Xi_{\varepsilon,q,\pi}(\calI)$ with $q\leq 16/\varepsilon$ such that
\begin{equation}\label{eq:val-plus-greater-val}
\val^+(\xi, \calI) \geq 
\val(\varphi, \calI).
\end{equation}
\end{lemma}
\begin{proof}
Let $T$ be the tree used in solution $\varphi$, and $\Lambda$ be the set of its leaves. Solution $\varphi$ maps the set of variables $V$ to $\Lambda$. We now define a function $\lambda:\Lambda \to \Lambda$ that maps all leaves of $T$ to at most $q$ distinct leaves of $T$. This function also assigns a color to every leaf in the image. Then, we define the coarse solution $\xi = \lambda \circ \varphi$. That is, $\xi$ uses the true solution $\varphi$ to map 
a variable $x$ to a leaf $l$ and then uses function $\lambda$ to assign $x$ one of $q$ chosen leaves of $T$.

\medskip

\noindent{\textbf{Algorithm.}} We describe an algorithm for finding function $\lambda$. The algorithm first assigns a label to every leaf $u$ of $T$ and a color to every label. Then, it maps each label to an arbitrary 
(e.g., the leftmost) leaf of $T$ that has that label. 

Our algorithm considers all nodes of the tree in the bottom-up order. Denote the subtree rooted at node $u$ by $T_u$. For every $u$, the algorithm either processes subtree $T_u$ and marks $T_u$ as processed or skips node $u$. It processes $u$ if one of the following conditions is met:
\begin{itemize}
\item $u$ is the root of $T$; or
\item both the left and right subtrees of $u$ 
contain at least one already processed node; or
\item the number of yet unlabelled leaves in $T_u$ is greater than $\varepsilon |V|/2$.
\end{itemize}
Note that the second item can be rephrased as follows: $u$ is the least common ancestor (LCA) of two already processed nodes.

To process a node $u$, the algorithm creates four new labels $LL_u, LR_u, RL_u, RR_u$ and assigns the same new color to all of them. It assigns the first two labels $LL_u$, $LR_u$ to leaves in the left subtree of $u$ and the second two labels $RL_u$, $RR_u$ to leaves in the right subtree of $u$. Consider the left subtree. If it does not contain already processed nodes, then all leaves of the tree receive label $LL_u$. Otherwise, there should be one processed node $v$ such that subtree $T_v$ contains all other processed nodes in the left subtree of $T_u$. This node $v$ is the least common ancestor (LCA) of all processed nodes in the left subtree of $T_u$. 
We assign label $LL_u$ to the leaves in the left subtree of $u$ that are to the left of $T_v$ and label $LR_u$ to 
the leaves in the left subtree of $u$ that are to the right of $T_v$. We assign labels in the right subtree in a similar way. See Figure~\ref{fig:coloring-algorithm} in the Appendix.

\medskip

\noindent\textbf{Value of the solution.}
We now show that for the coarse solution $\xi$ constructed above, inequality~(\ref{eq:val-plus-greater-val}) holds. Consider a constraint $(x_1,\dots,x_k)$ in $\calI$. We need to show that
$f^+(\xi(x_1),\dots,\xi(x_k)) \geq 
f(\xi(x_1),\dots,\xi(x_k)) $. If at least two variables $x_i$ and $x_j$ have the same color, then 
$f^+(\xi(x_1),\dots,\xi(x_k)) 
= 1$ and, therefore, $f^+(\xi(x_1),\dots,\xi(x_k))  = 1\geq 
f(\xi(x_1),\dots,\xi(x_k))$. So from now on, we assume that $x_1,\dots,x_k$ have distinct colors. 

Recall, that payoff function $f$ can be specified by a list of patterns and corresponding payoffs: function $f(y_1,\dots,y_k)$ returns a certain value if $y_1,\dots,y_k$ match the corresponding pattern.
Each pattern $P$ can be described either by a tree with $k$ leaves $l_1,\dots, l_k$ or as a conjunction of
\emph{bracket} predicates of the form
$[y_a < y_b]$,
$[y_a, y_b < y_c]$,
$[y_a < y_b, y_c]$. See Section~\ref{sec:cl:bracket-predicates} and Lemma~\ref{lem:bracket-predicates} for details.
In this proof, we will use the latter type of pattern descriptions.

\begin{claim}
If $\varphi(x_1),\dots, \varphi(x_m)$ match a pattern $P$, then $\xi(x_1),\dots, \xi(x_m)$ match the same pattern $P$. Here, $\varphi$ is the original solution, and $\xi$ is the corresponding coarse solution.
\end{claim}
\begin{proof}
Consider an arbitrary predicate $[y_a,y_b < y_c]$, which is a part of $P$. If $\varphi(x_1),\dots,\varphi(x_m)$match $P$, then the predicate  $[\varphi(x_a),\varphi(x_b) < \varphi(x_c)]$ must be true. We show that 
$[\xi(x_a),\xi(x_b) < 
\xi(x_c)]$ is also true.

Examine node $u$ in $T$ where $\varphi(x_a),\varphi(x_b),\varphi(x_c)$ are split into two groups. This node $u$ is the least common ancestor of $\varphi(x_a)$, 
$\varphi(x_b)$, $\varphi(x_c)$ in tree $T$.
Since $[\varphi(x_a),\varphi(x_b) < \varphi(x_c)]$ is true, $\varphi(x_a)$ and $\varphi(x_b)$ must belong to the left subtree of $u$; and $\varphi(x_c)$ must belong to the right subtree of $u$.
Now, there are two possibilities: $u$ was or was not processed by the algorithm. 

If $u$ was processed by the algorithm, then $\varphi(x_a), \varphi(x_b)$ received labels in the left subtree and $\varphi(x_c)$ received labels in the right subtree. In the coarse solution, labels in the left subtree are mapped to leaves in the left subtree, and labels in the right subtree mapped to leaves in the right subtree. Hence, the predicate $[\xi(x_a),\xi(x_b) < \xi(x_c)]$ is true in the coarse solution.

Suppose now that $u$ was not processed by the algorithm but, of course, one of its ancestors was processed. Denote the first ancestor of $u$ which was processed by $v$. Assume without loss of generality that $u$ is in the left subtree of $v$. When the algorithm processed $v$, it assigned two new labels $LL_v, LR_v$ to some leaves in $T_u$. 
These labels have the same color. Since $\varphi(x_a)$, $\varphi(x_b)$, $\varphi(x_c)$ have distinct colors, only one of them could have received label $LL_v$ or $LR_v$. Therefore, the other two leaves were assigned labels 
before $v$ was processed. Suppose they were assigned labels when $v'$ was processed. Note that $v'$ is a descendent of $v$ and $u$ (if $v'$ was on the path from $v$ to $u$, then $v'$ not $v$
would be the first ancestor of $u$ that was processed). If $v'$ belonged to the right subtree of $u$, only 
$\varphi(x_c)$ would be its descendant and, consequently, neither   $\varphi(x_a)$ nor $\varphi(x_b)$ would receive a label when $v'$ was processed. Hence, $v'$ is in the left subtree of $u$. 
Therefore, $\varphi(x_c)$ must have received label $LR_v$, and $x_a$, $x_b$ received labels in $T_{v'}$. Since all leaves having label $LR_v$ are to the right of leaves in subtree $T_{v'}$, we  get that 
$[\xi(x_a),\xi(x_b) < \xi(x_c)]$ is satisfied in the coarse solution.
\end{proof}
This completes the proof of bound~(\ref{eq:val-plus-greater-val}). It remains to bound the number of labels and colors used by the algorithm.

\medskip

\noindent\textbf{The size of the coarse solution.} We now show that $\xi\in\Xi_{\varepsilon,q,\pi}$ (see Section~\ref{sec:coarse-solution} for definition of
$\Xi_{\varepsilon,q,\pi}$). Consider the step of the algorithm when a node $u$ is processed. Observe that each label $LL_u$, $LR_u$, $RL_u$, $RR_u$ is assigned to consecutive leaves in $\pi = \order(\varphi)$. All of them have the same color and no other label has that color.  The number of unlabeled leaves in the left and right subtrees of $u$ is at most $\varepsilon |V|/2$. So, the total number of leaves that get colored at this step of the algorithm is at most  $\varepsilon |V|$.

We now estimate the number of leaves in the image of 
$\xi$. This number equals the number of labels we create, which, in turn, equals the number of processed vertices multiplied by $4$. Each node $u$  is processed by the algorithm because of one the three reasons provided in the definition of the algorithm. The number of nodes $u$ processed because $T_u$ has more than $\varepsilon |V|/2$ unlabelled leaves is at most
$|V|/(\varepsilon |V|/2) = 
2/\varepsilon$. The number of nodes $u$ with at least one processed node in both the left and right subtrees of $u$ is at most $2/\varepsilon - 1$. Additionally, the algorithm always processes the root of $T$. Thus, the total number of processed nodes and, consequently, number of colors is at most $4/\varepsilon$. The number of labels is upper bounded by $16/\varepsilon$.  
\end{proof}
\subsection{No Good Coarse Solution for Gap Instance}\label{sec:no-good-cs-for-gap}

In the previous two sections, we proved that the gap instance $\calI^f_{gap}$ has a solution $\varphi$ of value $1$ and then showed how to transform every true solution to a coarse solution 
$\xi$ with $\val^+(\xi,\calI) \geq \val(\varphi,\calI)$. Thus, we know that for every $\varepsilon$ and $q$, there exists an ordering $\pi$ and coarse solution $\xi\in \Xi_{\varepsilon,q,\pi}(\calI)$ with $\val^+(\xi, \calI_{gap}^f) = 1$. We now prove that, in conrast, $\val^-(\xi, \calI_{gap}^f) \leq \alpha + \varepsilon$
for every 
$\xi\in \Xi_{\varepsilon,q,\pi}(\calI)$ 
if the gap instance is sufficiently large.

\begin{lemma}\label{lem:no-good-cs-for-gap}
For every positive $k,q\in \bbN$ and $\varepsilon'\in(0,1)$, there exists 
$m^*$ such that the following claim holds. For every phylogenetic payoff function $f$ of arity $k$, gap instance $\calI_{gap}^f=(V,C)$ with $|V|\geq m^*$, and coarse solution $\xi \in \Xi_{\varepsilon,q,\pi}$,
we have:
\begin{equation}\label{eq:lem:no-good-cs-for-gap}
\val^-(\xi,\calI))\leq \alpha + \varepsilon',
\end{equation}
where $\alpha$ is the biased random assignment threshold for payoff function $f$; $\varepsilon$ and $\pi$ are arbitrary.
\end{lemma}
\begin{proof}
Consider a coarse solution $\xi$. It maps variables of instance $\calI_{gap}^f$ to leaves of some tree $T$. Since $\xi\in\Xi_{\varepsilon,q,\pi}$, tree $T$ has at most $q$ leaves $l_1,\dots,l_{q'}$. We view this coarse solution as a solution to an \emph{ordinary} CSP with alphabet 
$l_1,\dots,l_{q'}$ and payoff function $f^-$. This function applied to variables $x_1,\dots, x_k$ returns $f(x_1,\dots,x_k)$ if all colors assigned to $x_1,\dots,x_k$ are distinct and $0$, otherwise. By Lemma~\ref{lem:ordinary-CSPs-gap}, the value of this solution is at most 
$\alpha'+\varepsilon'$, where $\alpha'$ is the expected value of the optimal biased random assignment for payoff function $f^-$. 

To complete the proof, we show that $\alpha'\leq \alpha$. Consider a biased random assignment algorithm
with some probability distribution $\rho$ on labels $l_1,\dots, l_{q'}$. We can use this distribution to define a biased randomized algorithm for phylogenetic CSP instance $\calI_{gap}^f$. The biased assignment algorithm first randomly and independently assigns all vertices $V$ to leaves $l_1,\dots, l_{q'}$ with probabilities
$\rho(l_1),\dots,\rho(l_{q'})$. Then, it recursively partitions vertices assigned to each leaf each time splitting vertices between the left and right subtrees with probability $50\%/50\%$ (see Section~\ref{sec:random-assignment-def} for details). Note that the expected value of this randomized algorithm for phylogenetic payoff function $f$ is greater than or equal to the value of the \emph{ordinary} payoff function $f^-$. This is the case because both payoff functions have the same value if the colors of the leaves assigned to their arguments are distinct.
However, $f^-(x_1,\dots,x_k) = 0$ if the colors of two or more leaves $x_i$ and $x_j$ are the same. This implies that 
$$
\E_{x_i\sim \rho}[
f^-(x_1,\dots,x_k)]\leq
\E_{x_i\sim \rho}[
f(x_1,\dots,x_k)].$$
\end{proof}
\subsection{Coarse Solutions for Random Orderings}\label{sec:diff-val-pm}
In this section, we bound the maximum difference 
$\max_{\xi\in \Xi_{\varepsilon,q,\pi}}
(\val^+(\xi,\calI) - 
\val^-(\xi,\calI))$ for a random ordering $\pi$. We assume that the instance $\calI$ of phylogenetic CSP is regular. That is, the weight of constraints that contain  a variable $x$ is the same for all $x\in V$. Note that our gap instance $\calI^f_{gap}$ satisfies this condition.
In the lemma below, we will use the notion of a Gaifman graph. The Gaifman graph for  instance $\calI$ of a constraint satisfaction problem is a weighted graph on the variables $V$ of $\calI$. The weight of edge $(x_1,x_2)$ equals the total weight of all constraints that depend on $x_1$ and $x_2$. Given an instance $\calI$, we construct the Gaifman graph for $\calI$ as follows.  For every constraint $(x_1,\dots,x_k)$, we  
add a clique on $x_1,\dots,x_k$ with the weight of edges equal to the weight of constraint $x_1,\dots, x_k$. It is easy to see that if $\calI$ is a regular instance (see above), then its
Gaifman graph is also regular.

Let $H=(V,E)$ be the weighted Gaifman graph for a regular instance $\calI$. Consider an arbitrary coarse solution $\xi\in \Xi_{\varepsilon,q,\pi}(\calI)$. Denote the total weight of monochromatic edges in $H$ by $\mc(\xi,H)$:
$$\mc(\xi,H) = 
\weight(\{(x,y)\in E: \clr(\xi(x))=
\clr(\xi(y))\}).$$
We show $\mc(\xi,H)$ is small on average for a random ordering $\pi$.
\begin{lemma}\label{lem:bound-max-monochrome}
For every $\varepsilon \in (0,1)$ and 
positive $q\in \bbN$, there exists 
$m^*=O(q\log(q/\varepsilon)/\varepsilon^2)$ such that for every regular instance $\calI=(V,C)$ with $|V|\geq m^*$, the following bound holds:
\begin{equation}\label{eq:lem:bound-max-monochrome}
\E_{\pi}\Big[\max_{\xi\in \Xi_{\varepsilon,q,\pi}(\calI)} \mc(\xi,H)\Big]\leq 3\varepsilon\cdot \mathrm{weight}(E).
\end{equation}
Here, $H=(V,E)$ is the Gaifman graph of $\calI$; $\pi$ is a random ordering of $V$.
\end{lemma}
\begin{proof}
Let $m=|V|$. We rescale the weight of all edges so that the weight of edges leaving any node in $H$ equals $2/m$. Then, the total weight of all edges in $H$ is $1$. In this proof, we will ignore the tree structure of the coarse solution. The ordering $\pi$ is a one-to-one mapping of $V$ to $\{0,\dots, m-1\}$. Thus, the coarse solution $\xi$ defines a coloring $\chi$ on $\{0,\dots, m-1\}$: The color of $i$ equals the color assigned by $\xi$ to the preimage of $i$. That is,
$$
\chi(i) = \clr(\xi(\pi^{-1}(i))).
$$
Note that (1) $\xi$ assigns the same color to at most $\varepsilon m$ numbers; (2) the entire set $\{0,\dots,m-1\}$ is partitioned in 
at most $q$ groups of consecutive  numbers and each group receives some color (every consecutive group corresponds to a leaf in the coarse solution; different groups may have the same color). We now rephrase the statement of the lemma as follows:
\begin{equation}\label{eq:lem:bound-max-monochrome-prime}
\E_{\pi}\Big[\max_{\chi} \mc(\chi\circ\pi ,H)\Big]\leq 3\varepsilon,    
\end{equation}
where $\chi$ is a coloring satisfying the conditions above; and
$\mc(\chi\circ\pi, H)$ is the fraction of monochromatic edges in $H$ with respect to the coloring $\chi\circ\pi$.

The proof follows a standard probabilistic argument. First, we estimate the number of monochromatic edges for a fixed coloring $\chi$ and random permutation $\pi$. Specifically, we argue that for a fixed coloring, the expected weight of monochromatic edges is at most $\varepsilon$. Then, we use Maurey’s concentration inequality for
permutations to show that for a typical permutation $\pi$, the maximum
number of monochromatic edges $\mc(\chi\circ \pi,H)$ over all colorings $\chi$ is at most $2\varepsilon$.  This yields the desired bound~(\ref{eq:lem:bound-max-monochrome-prime}).

Consider a fixed coloring $\chi$. By the definition, it assigns each color to at most $\varepsilon m$ numbers. Let us now orient all edges of $H$ in an arbitrary way. The probability that the right endpoint of an edge is assigned the number of the same color as the left endpoint is at most $\varepsilon$. Thus, the expected weight of monochromatic edges is at most $\varepsilon$.

The total number of different colorings satisfying conditions (1), and (2) above is at most $(qm)^q$, because we can specify the leftmost number in each group in at most $m^q$ ways; we can then assign colors to these $q$ groups in at most $q^q$ ways.

We will now use Maurey's concentration inequality
(\cite{maurey}; see also
Theorem 5.2.6 in the book by~\cite{vershynin} and
Theorem~13 in lecture notes by~\cite{naor-concentration})
to bound the probability that for a random $\pi$ the weight of monochromatic edges is greater than $2\varepsilon$. To this end, define the distance between two permutations or orderings $\pi'$ and $\pi''$ as the fraction of $x$ where $\pi'$ and $\pi''$ differ:
$$
\dist(\pi',\pi'') = \frac{|\{x\in V: \pi'(x)\neq \pi''(x)\}|}{|V|}.
$$
A function $f:Sym(m)\to \bbR$ is $L$-Lipschitz if 
$f(\pi') - f(\pi'')\leq L \cdot \dist(\pi',\pi'')
$
for all permutations $\pi',\pi'' \in Sym(m)$. Here, $Sym(m)$ is the group of all permutations on $m$ elements (symmetric group). 
\begin{theorem}[Maurey]\label{thm:Maurey} Consider an $L$-Lipschitz  function $f:Sym(m)\to\bbR$. Let $\pi$ be a random permutation in $Sym(m)$. Then,
\begin{equation}
\label{eq:thm:Maurey}
\Pr\{|f(\pi) - \E[f]|\geq t\}
\leq 2e^{-\frac{c t^2m}{L^2}}
\end{equation}
for some constant $c > 0$.
\end{theorem}
We will apply Theorem~\ref{thm:Maurey} to the function $\pi\mapsto \mc(\chi\circ \pi, H)$. To do so, we need the following claim.

\begin{claim}\label{clm:Lip}
The function 
$\pi\mapsto \mc(\chi\circ \pi, H)$ is 2-Lipschitz. 
\end{claim}
\begin{proof}
Consider two orderings $\pi'$ and $\pi''$. We split all edges of $G$ into two sets $A$ and $B$. Set $A$ contains edges $(x,y)$ with 
$\pi'(x) = \pi''(x)$ and $\pi'(y) = \pi''(y)$. Set $B$ contains the remaining edges. Each edge with both endpoints in $A$ is assigned the same colors by $\chi\circ\pi'$ and $\chi\circ\pi''$. Thus, it is either monochromatic with respect to both colorings $\chi\circ\pi'$ and $\chi\circ\pi''$, or not monochromatic with respect to both colorings $\chi\circ\pi'$ and $\chi\circ\pi''$. Hence, $\mc(\chi\circ \pi', H) - 
\mc(\chi\circ \pi'', H)\leq \textrm{weight}(B)$. However, each edge in $B$ is incident on a node $x$ with $\pi'(x)\neq \pi''(x)$. Since the number of such nodes equals $\dist(\pi',\pi'')\,|V|$,
the total weight of edges in $B$ is at most $2\dist(\pi',\pi'')$. Here, we use that the total weight of edges incident on any fixed vertex $x$ is $2/|V|$. 
This concludes the proof of Claim~\ref{clm:Lip}.
\end{proof}

By Maurey's concentration inequality~(\ref{eq:thm:Maurey}), we have
$$\Pr_{\pi}\Big\{\mc(\xi\circ\pi,H) - 
\varepsilon\geq \varepsilon
\Big\}
\leq 2e^{-c' \varepsilon^2 m}$$
for some positive constant $c'$. We now apply the union bound over all possible colorings 
$\chi$ and the the following inequality:
$$\Pr_{\pi}\{\max_{\chi} \mc(H,\chi\circ\pi) \geq 2\varepsilon\}
\leq 
2e^{-c' \varepsilon^2 m} q^qm^q.$$
If $m>C'q\log(q/\varepsilon)/\varepsilon^2$ (for sufficiently large constant $C'$), then the right hand side of the inequality is less than $\varepsilon$. Since $\mc(\chi\circ\pi)$ is upper bounded by the total weight of all edges, which is $1$, we get the desired bound (\ref{eq:lem:bound-max-monochrome-prime}).
\end{proof}

We use Lemma~\ref{lem:bound-max-monochrome} to bound $\max_{\xi\in \Xi_{\varepsilon,q,\pi}}
(\val^+(\xi,\calI) - 
\val^-(\xi,\calI))$.

\begin{lemma}\label{lem:diff-val-pm}
For every $\varepsilon \in (0,1)$ and 
positive $k,q\in \bbN$, there exists 
$m^*=O(q k^2\log(kq/\varepsilon)/\varepsilon^2))$ such that for every regular instance $\calI$ with at least $m^*$ variables, the following bound holds:
\begin{equation}\label{eq:lem:diff-val-pm}
\E_{\pi}\Big[
\max_{\xi\in \Xi_{\varepsilon,q,\pi}}
(\val^+(\xi,\calI) - 
\val^-(\xi,\calI))\Big]\leq \varepsilon.
\end{equation}
Here, $\pi$ is a random ordering of variables $V$ of the instance $\calI$.
\end{lemma}
\begin{proof} 
Let $H$ be the Gaifman graph for instance $\calI$. Consider a coarse solution $\xi$ and the induced
coloring of variables $V$.
Let us now examine the definition of functions $\val^+$ and 
$\val^-$ given in Equations~(\ref{eq:def:val-pm-1}) and (\ref{eq:def:val-pm-2}). These functions differ only on payoff functions $f(\xi(x_1),\dots,\xi(x_k))$ with two variables having the same color
(i.e., $\clr(\xi(x_i)) = \clr(\xi(x_j))$). Thus, 
$$\val^+(\xi,\calI) - 
\val^-(\xi,\calI)
\leq \mc(\xi,H).
$$
Note that the total weight of all constraints in the instance $\calI$ is $1$; and the total weight of all edges in $H$ is $k(k-1)/2$, because for every payoff function in $\calI$, we create a clique of size $k$  in $H$. We now apply Lemma~\ref{lem:bound-max-monochrome} with $\varepsilon'=2\varepsilon/(3k(k-1))$ and get inequality~(\ref{eq:lem:diff-val-pm}).
\end{proof}

\subsection{Proof of Theorem~\ref{thm:exists-gap-instance}}
We now complete the proof of Theorem~\ref{thm:exists-gap-instance}. 
Consider a payoff function $f$ of arity $k$. We assume that the arguments of $f$ are rearranged so that there exists an assignment $\varphi$ to the variables that satisfies $f$ and such that the ordering of variables in $\varphi$ is $x_1,\dots,x_m$ (i.e., 
$\order(\varphi) = id$). By Corollary~\ref{cor:gap-sat}, the gap instance $\calI^f_{phy}$ is completely satisfiable. 
Suppose that the number of  leaves in $\calI^f_{gap}$ is larger than some sufficiently large $m^*$. By Lemma~\ref{lem:good-cs}, we have
$$\E_{\pi}
\Big[
\opt(\calI^f_{gap} \mid \pi)
\Big]
=
\E_{\pi}
\Big[
\max_{\varphi\in\Phi_{\pi}}
(\val(\varphi, \calI))
\Big]
\leq
\E_{\pi}
\Big[
\max_{\xi\in\Xi_{\varepsilon,q,\pi}}
(\val^+(\xi, \calI))
\Big]
.$$
By Lemma~\ref{lem:diff-val-pm},
$$\E_{\pi}
\Big[
\max_{\xi\in\Xi_{\varepsilon,q,\pi}}
(\val^+(\xi, \calI))
\Big]
\leq 
\E_{\pi}
\Big[
\max_{\xi\in\Xi_{\varepsilon,q,\pi}}
(\val^-(\xi, \calI))
\Big]+\varepsilon.
$$
Finally, Lemma~\ref{lem:no-good-cs-for-gap},
$$\E_{\pi}
\Big[
\max_{\xi\in\Xi_{\varepsilon,q,\pi}}
(\val^-(\xi, \calI))
\Big]\leq \alpha+\varepsilon.
$$
This concludes the proof of Theorem~\ref{thm:exists-gap-instance}.
\section{Making the Reduction from Unique Games Work}\label{sec:reduction-ug}
We now examine the hardness reduction by
\citet*{GHMRC11} and then modify it to make it work with our own reduction $h_{ord\to phy}$. As most other hardness reductions from the Unique Games Conjecture, the hardness reduction by \citet{GHMRC11} relies on a
\emph{dictatorship test} for the problem (see~\cite*{khot2002power, khot2003vertex,  kkmo,raghavendra2008}). A dictatorship test is a special instance of the problem, in our case ordering CSP $\Gamma_{ord}$, on variables in the grid $[M]^R$. On the one hand, this instance must have  a \emph{dictator} solution $\varphi$ of value at least $1-\varepsilon$. On the other hand, every $\tau$-pseudorandom solution for this instance must have value at most $\alpha + \varepsilon$, where $\alpha$ is the desired approximation hardness of the problem.  We remind the reader that a dictator is a function $\varphi$ defined on $\mathbf{z}\in[M]^R$ that depends only on one coordinate $\mathbf{z}_j$ of $\mathbf{z}$. That coordinate $j$ is \emph{the dictator}. A function $\varphi$ is $\tau$-pseudorandom if $T_{1-\varepsilon}\varphi$ does not have influential coordinates i.e., coordinates with influence greater than $\tau$ 
(here $T_{1-\varepsilon}$ is the noise operator that ``flips'' every coordinate of $\mathbf{z}$ with probability $\varepsilon$). Note that we can pick a constant $M$ as we wish (it can depend on $\varepsilon$, which we treat as a fixed constant). However, the value of $R$ depends on the Unique Games instance we use in the reduction and is not under our control ($R$ equals the number of labels in the Unique Games instance).

The general recipe for creating dictatorship tests was provided
by~\cite*{raghavendra2008} in his influential paper on optimal approximation algorithms and approximation hardness for ordinary CSPs. His dictatorship test was adapted for ordering CSPs by~\cite{GHMRC11}. Also, ~\cite{GHMRC11} defined $\tau$-pseudorandom functions for ordering CSPs (see Definition 4.2 in their paper) and developed tools necessary for analyzing such functions.

We now outline the dictatorship test used by~\cite{GHMRC11}. We will work with the ordering predicate $o$ of arity $m$  defined in Section~\ref{sec:proof-overview}.
\cite{GHMRC11} use a gap instance 
$\widetilde \calI_{gap}^o$ with $M$ variables. This is the same gap instance\footnote{For technical reasons, they take several copies of this gap instance. However, in our modified hardness reduction, we will use this instance $\widetilde\calI^o_{gap}$ as is.}
as we described in the proof in Section~\ref{sec:proof-overview} only of a larger size and applied to the ordering predicate $o$. Then, for every tuple $\mathbf{s}$ of $m$ variables 
$s_1,\dots, s_m\in [M]$ ($m$ is the arity of the ordering CSP), they 
define a random map $L_{\mathbf{s}}$ that maps $s_1,\dots, s_m$ to another $m$ tuple in $[N]^k$ (in their case $N=M$). This map should satisfy several important conditions we examine in a moment. We give a description of the dictatorship test in Figure~\ref{fig:dict-test}.
\begin{figure}[ht]    
\centering
\begin{tcolorbox}
Dictatorship Test from~\cite{GHMRC11}:
\begin{itemize}
\item[-] Pick a random constraint $(s_1,\dots,s_m)$ from  $\widetilde \calI_{gap}^o$.
\item[-] Draw $m$ vectors
$\mathbf{z}_1,\dots, \mathbf{z}_m\in[N]^R$ using $R$ independent random functions 
$L_{\mathbf{s}}^{(1)},\dots, L_{\mathbf{s}}^{(R)}$:
$$
(\mathbf{z}^{(j)}_1,\dots,
\mathbf{z}^{(j)}_m) = 
L_{\mathbf{s}}^{(j)}(s_1,\dots,s_m).
$$
\item[-] Apply $\varepsilon$-noise to each $\mathbf{z}_i^{(j)}$ i.e. with probability $\varepsilon$, replace it with a random value in~$[N]$.
\item[-] Return constraint $(\mathbf{z}_1,\dots,\mathbf{z}_m)$ for the payoff function $o$.
\end{itemize}

The instance $\calI$ of the ordering CSP $\Gamma_{ord}$ generated by the dictatorship test consists of a set of variables $V = [N]^R$ and set of constraints $C= \underbrace{V \times\dots \times V}_{m}$. The weight of each constraint 
$(\mathbf{z}_1,\dots,\mathbf{z}_m)$
equals the probability that this constraint is returned by the procedure above.
\end{tcolorbox}
\caption{Dictatorship Test}\label{fig:dict-test}
\end{figure}

The map $L_{\mathbf{s}}$ should satisfy several conditions. First, for every $j$, the dictatorship solution 
$\varphi: \mathbf{z}\mapsto \mathbf{z}_j$ (where $\mathbf{z}\in V= [M]^R$ is a variable; $\mathbf{z}_j$ is a number in $[N]$)
should have value $1-O(\varepsilon)$. In this solution, several distinct variables $\mathbf{z}$ can be mapped to the same position $i$; in this case, we pick a random ordering among them, but preserve their relative order with other $\mathbf{z}$'s.
Then, $L_{\mathbf{s}}$ should be $\eta$-\emph{smooth} i.e., for all $t_1,\dots,t_m \in [N]^R$, we have (for some $\eta > 0$)
\begin{equation}\label{eq:UG:L-smooth}
\Pr\big\{L_{\mathbf{s}}(s_1,\dots,s_m) = (t_1,\dots,t_m)\big\} \geq \eta.
\end{equation}
The marginal distribution of each coordinate of $L_\mathbf{s}$ should be uniform i.e., for every $\mathbf{s}\in [M]^m$, every $i\in \{1,\dots, m\}$, and every $t\in[N]$,
\begin{equation}\label{eq:UG:L-uniform}
\Pr\{L_{\mathbf{s}}(\mathbf{s})_i= t\} = \frac{1}{N}.
\end{equation}
Here, $L_{\mathbf{s}}(\mathbf{s})_i$ denotes the $i$-th coordinate of $L_{\mathbf{s}}(\mathbf{s})$. Finally, there should exist a \emph{global} SDP solution (the same for all functions $L_{\mathbf{s}}$) that match the first and second moments of every $L_{\mathbf{s}}$.

Unfortunately, this dictatorship test instance does not work for us as is. As we discussed earlier, we need to get a hardness reduction 
$h_{UG\to ord}$ 
from Unique Games to ordering CSP $\Gamma_{ord}$, which not only maps almost satisfiable instances of Unique Games to almost satisfiable instances of $\Gamma_{ord}$, but also satisfies the following condition: The composition of hardness reductions
$h_{ord\to phy} \circ h_{UG\to ord}$ maps almost satisfiable instances of Unique Games to almost satisfiable instances of our phylogenetic CSP $\Gamma_{phy}$. To satisfy this condition, we need map $L_{\mathbf{s}}$ to have one additional property.
Each dictatorship solution 
$\varphi:\mathbf{z}\mapsto \mathbf{z}_j$ must have value at least $1-O(\varepsilon)$ when evaluated on the phylogenetic CSP corresponding to the dictatorship test instance (i.e., the image of the dictatorship test instance under $h_{ord\to phy}$). Note that each $\mathbf{z}_j$ is a leaf in the tree must be a leaf of a tree, so we also need to provide a tree whose leaves are elements of $[N]$.

The map used in the paper by \cite{GHMRC11} cyclically shifts elements in $[N]$ (in their case, $N=M$). This destroys any tree structure we can define on $[N]$. Let us illustrate this point by example. Consider the Triplet Consistency constraint $uv|w$ and binary tree of depth 2 with $4$ leaves $1$, $2$, $3$, $4$. This constraint is satisfied if $u=1$, $v=2$, $w=3$. However, if we shift values by one, $u=2$, $v=3$, $w=4$, then the constraint is no longer satisfied.

We are going to define an alternative random function $L$ that maps all variables in $[M]$ to some larger domain $[N]$. The elements of $[N]$ are associated with leaves of a binary tree. We then let $L_{\mathbf{s}}(s_1,\dots, s_k) = (L(s_1),\dots, L(s_k))$ and plug these functions $L_{\mathbf{s}}$ into the dictatorship test described above.

To make the proof of \cite{GHMRC11} work for this new function $L$, we need to ensure that maps $L_{\mathbf{s}}$ satisfy the required conditions. Finding a global SDP solution for $L$ is easy: We get it for free, because $L$ is a global distribution and, as such, is a convex combination of integral solutions (each realization of $L$ is an integral solution; it maps variables in $[M]$ to leaves in $[N]$). The smoothness condition~(\ref{eq:UG:L-smooth}) can be easily obtained by perturbing  $L$. 

Now, we show that there exists a random function $L:[M]\to [N]$
that satisfies the following conditions:
\begin{itemize}
\item for all $u\in [M]$ and $v\in [N]$, $\Pr\{L(u) = v\} = 1/N$ (cf. Equation (\ref{eq:UG:L-uniform}));
\item for every $j$ and assignment $\varphi: \mathbf{z} \mapsto\mathbf{z}_j$, we have
$\val(\varphi, h_{ord\to phy}(\calI_{test}))\geq 1 - O(\varepsilon)$, where $\calI_{test}$ is the dictatorship test instance obtained using function $L$.
\end{itemize}
Let us examine the second condition. Denote 
$\calI_{red} = h_{ord\to phy}(\calI_{test})$. Recall that instance $\calI_{red}$ is obtained from the dictatorship test instance $\calI_{test}$ by replacing every constraint $(\mathbf{z}_1,\dots,\mathbf{z}_m)$ for the payoff function $o$ with a copy of the phylogenetic gap instance $\calI^f_{gap}$ on the same set of variables $(\mathbf{z}_1,\dots,\mathbf{z}_m)$. We remind the reader that $m$ is the number of leaves in the gap instance 
$\calI_{gap}^f$, and hence is a power of $k$. Similarly, $M$ is the number of leaves in the gap instance 
$\widetilde \calI_{gap}^o$ and is a power of $m$, and, consequently, a power of $k$. We let $N = M^d$ for some constant $d$. Thus, $m$, $M$, and $N$ are powers of $k$. We will associate sets $[m]$, $[M]$, and $[N]$ with leaves of $k$-ary trees of appropriate depths. We will also map set $[N]$ to leaves of a binary tree using  Lemma~\ref{lem:cousins-sat}. We will use this mapping to define $\val(\varphi,\calI_{red})$. We now show how to construct the desired function $L$ for a sufficiently large number $N$.
Later, in Lemma~\ref{lem:L-works}, we will prove that L satisfies the required conditions.

\begin{lemma}\label{lem:exists-L}
Fix a natural $k > 1$ and consider a 
perfect $k$-ary tree $T_M$ with $M$ leaves labeled $0,\ldots, M-1$. For every positive $\varepsilon$, there exists an integer $N$ and a random map $L$ from leaves of $T_M$ to leaves of another $k$-ary tree $T_N$ with $N$ leaves  labeled $0,\ldots, N-1$ such that
\begin{itemize}
\item for every $u \in [M]$ and $v\in [N]$, we have
$$\Pr\{L(u) = v\} = 1/N;$$
\item for every $k$ cousins $u_1,\dots,u_k$ in $T_M$, 
$$\Pr\big\{\cousins(L(u_1),\cdots, L(u_k))\big\} \geq 1 - O(\varepsilon);$$
$$\Pr\{L(u_i) = v_i\;\forall i\} > 0.$$
\end{itemize}
\end{lemma}
\noindent\textbf{Remark:} We define the notion of \emph{cousins} in Section~\ref{sec:gap-sat}. Leaves $u_1,\dots,u_k$ are cousins in a tree of arity $k$ if each $u_i$ lies in the subtree rooted at the $i$-th child of $\lca(u_1,\dots,u_k)$.
\begin{proof}
Let $N=M^d$ for $d=\frac{3M}{\varepsilon^2}\ln (\frac{M}{\varepsilon})$. We create $k$-ary tree $T_N$ with $N$ leaves $0,\dots, N-1$. We also define a set of ``shortcut'' edges for $T_N$. These edges go from level $0$ to $d'$, $d'$ to $2d'$ and so on, where $d'=\log_{k} M$ is the depth of tree $T_M$.
We will denote the tree with shortcut edges by $T_{sc}$. This tree has arity $M$.

Consider the random map $L_{M,N}$ defined in Section~\ref{sec:proof-overview}. It maps $[M]$ to $[N]$. Note that it always maps leaves that are cousins  in $T_M$ to leaves that are cousins in $T_{N}$. We define $L$ using the following well-known lemma about the optimal coupling of random variables.

\begin{lemma}[Coupling Lemma; see e.g. \cite{roch}]\label{lem:opt-coupling}
Consider two probability distributions $\calP$ and $\calQ$ on a finite domain. Suppose random variable $X$ has distribution $\calP$, then there exists another random variable $Y$ having distribution $\calQ$ such that 
$$\Pr\{X\neq Y\} = \|\calP - \calQ\|_{TV},$$
where $\|\calP - \calQ\|_{TV}$ is the total variation distance between $\calP$ and $\calQ$.
\end{lemma}

The random variable $Y$ in  Lemma~\ref{lem:opt-coupling} can be obtained from $X$ using the maximum matching between distributions $\calP$ and $\calQ$. For each $u\in [M]$, we use this lemma to find a random variable $L(u)$ uniformly distributed in $[N]$
such that 
\begin{equation}\label{eq:dist-to-uniform}
\Pr\{L(u) \neq L_{M,N}(u)\} =
\frac{1}{2}\sum_{v\in[N]}\Big|\Pr\{L_{M,N}(u) = v\} - \frac{1}{N}\Big|.
\end{equation}
The expression on the right hand side is the total variation distance between the distribution of $L_{M,N}$ and the uniform distribution on $[N]$. We now upper bound this distance.

\begin{claim} For all $j\in[M]$,
$$\frac{1}{2}\sum_{v\in[N]}\Big|\Pr\{L_{M,N}(j) = v\} - \frac{1}{N}\Big|\leq \varepsilon.$$
\end{claim}
\begin{proof}
Consider a leaf $v$ in $M$-ary tree $T_{sc}$. Let $v(0),\dots, v(d) = v$ be the path from the root of the tree to $v$. For every $j\in [M]$, we count the number of times this path goes along the $j$-th branch of the tree. Namely, we let $B(v,j)$ be the number of nodes $v(i)$ such that $v(i+1)$ is the $j$-th child of $v(i)$. 

Recall that random function $L_{M,N}$ picks a random $t\in{0,\dots, d-1}$ and then selects a random node $u$ at depth $t$ in tree $T_{sc}$. If for this random $t$, $v(t+1)$ is the $j$-th child of $v(t)$, then $\Pr\{L_{M,N}(j) = v\mid t\} = M/N$ (because, in this case, $L_{M,N}(j) = v$ if two events occur: $u=v(t)$, and $v$ is randomly chosen in the subtree rooted at $v(t+1)$). Otherwise, $\Pr\{L_{M,N}(j) = v\mid t\} = 0$. Hence,
$$\Pr\{L_{M,N}(j) = v\} = \frac{M}{N}\cdot\frac{B(v,j)}{d}.$$
We have
$$\sum_{v\in[N]}\Big|\Pr\{L_{M,N}(j) = v\} - \frac{1}{N}\Big| = 
\sum_{v\in[N]}
\Big|\frac{M}{N}\cdot\frac{B(v,j)}{d}-\frac{1}{N}\Big|=
\frac{M}{d} \E_{v\in[N]}
\Big|B(v,j)-\frac{d}{M}\Big|
.$$
Consider a random leaf $v$ of $T_N$. The path from the root to $v$ is a random path. Every next vertex on this path is randomly chosen among the children of the current vertex. Thus, the probability that $v(i+1)$ is the $j$-th child of $v(i)$ is $1/M$ for every $i$. Consequently, $B_{v,j}$ is the sum of $d$ independent Bernoulli random variables with parameter $1/M$. By the Chernoff bound, 
$$
\Pr\Big\{\big|B(v,j)-\frac{d}{M}\big|\geq \varepsilon\cdot \frac{d}{M}\Big\}
\leq 
2e^{-\frac{\varepsilon^2(d/M)}{3}} < \frac{\varepsilon}{M}.
$$
Inequality $|B(v,j) - d/M| < d$ holds always. Thus,
$$
\E_{v\in[N]}
\Big|B(v,j)-\frac{d}{M}\Big|\leq
\varepsilon\cdot \frac{d}{M} + 
\frac{\varepsilon}{M} \cdot d = \frac{2d}{M}\varepsilon.
$$
\end{proof}

We now finish proof of Lemma~\ref{lem:exists-L}. Random function $L$ satisfies the first condition of Lemma~\ref{lem:exists-L} because each $L(u)$ is a random variable with uniform distribution in $[N]$. Function $L_{M,N}$ maps every set of cousins in $T_{N}$ to cousins in $T_{M}$. Thus, for all cousins $u_1,\dots u_k$ in $T_M$, we have
$$\Pr\big\{\cousins(L(u_1),\cdots, L(u_k))\big\} \geq
\Pr\big\{\cousins(L_{M,N}(u_1),\cdots, L_{M,N}(u_k))\big\} - \varepsilon k =
1 - \varepsilon k.$$
Here, we used that $\Pr\{L_{M,N}(u) \neq L(u)\} \leq \varepsilon$ for all $u\in[M]$.
This proves the second condition of function $L$ and completes the proof of Lemma~\ref{lem:exists-L}.
\end{proof}

We now verify that the random function $L$ from the previous lemma can be used in the dictatorship test. Specifically, we prove that the second item of Lemma~\ref{lem:exists-L} guarantees that $\val(\varphi,\calI_{red})\geq 1 - O(\varepsilon)$. After that, we smooth $L$ and plug it into the dictatorship test. The smooth variant of $L$ returns a completely random mapping into $[N]$ with a small probability $\eta'$ and with the remaining probability $1-\eta'$, it returns $L$.

\begin{lemma}\label{lem:L-works}
Let $L$ be the  random map $L:[M]\to [N]$ from Lemma~\ref{lem:exists-L}. Consider the dictatorship test instance
$\calI_{test}$  constructed using $L$. Let $\calI_{test} = h_{ord\to phy}(\calI_{test})$ be the corresponding instance of phylogenetic CSP. 
Finally, let $\varphi$ be a solution defined as $\varphi:\mathbf{z}\mapsto \mathbf{z}_j$. 
Then,
$$\val(\varphi,\calI_{red})\geq 
\min_{
\substack{
a_1,\dots,a_k\in [M]\\
\cousins(a_1,\dots,a_k) = 1}}
\Pr\big\{\cousins(L(a_1),\dots, L(a_k)) = 1\big\} - \varepsilon k.$$
\end{lemma}
\begin{proof}
Observe that the value of solution $\varphi$ equals
$$
\val(\varphi,\calI_{red}) = 
\E_{\mathbf{z}_1,\dots,\mathbf{z}_m}\E_{(i_1,\dots,i_k)} f(\mathbf{z}^{j}_{i_1},\dots, \mathbf{z}^{j}_{i_k}),
$$
where $(\mathbf{z}_1,\dots,\mathbf{z}_m)$ is a random constraint for the the ordering payoff function $o$ returned by the dictatorship test; and $(i_1,\dots,i_k)$ is a random constraint in the 
copy of $\calI_{gap}^f$ created by the reduction 
$h_{ord\to phy}$ for 
constraint $(\mathbf{z}_1,\dots,\mathbf{z}_m)$. 

\medskip

The probability that one of 
$\mathbf{z}^{j}_{i_1},\dots, \mathbf{z}^{j}_{i_k}$ is affected by $\varepsilon$-noise and replaced by a random value at the third step of the dictatorship test is at most $\varepsilon k$, since each of the values is changed with probability at most $\varepsilon$. If $\mathbf{z}^{j}_{i_1},\dots, \mathbf{z}^{j}_{i_k}$ are not changed at the third step, then
$$(\mathbf{z}^{j}_{i_1},\dots, \mathbf{z}^{j}_{i_k}) = 
(L^{(j)}(s_{i_1}),\dots,L^{(j)}(s_{i_k})),$$
where $(s_1,\dots, s_m)$ is a random constraint 
in $\tilde{\calI}^o_{gap}$
selected at the first step of the dictatorship test. Consequently,
$$\val(\varphi,\calI_{red})\geq
\E[f(L^{(j)}(s_{i_1}),\dots,L^{(j)}(s_{i_k}))] - \varepsilon k,
$$
where the expectation on the right hand side is taken over the random choice of 
$s_1,\dots s_m$, $i_1,\dots, i_k$, and random realization of $L^{(j)}$. Variables in every constraint in the gap instance 
$\calI_{gap}^f$ are cousins; $(s_{i_1},\dots,s_{i_k})$ is a constraint in the copy of $\calI_{gap}^f$
created for constraint
$(s_1,\dots, s_m)$.
Thus, $s_{i_1},\dots,s_{i_k}$ are also \emph{cousins}. Hence,
$$
\val(\varphi,\calI_{red})\geq
\min_{
\substack{
a_1,\dots,a_k\in [M]\\
\cousins(a_1,\dots,a_k) = 1}}
\E[f(L^{(j)}(a_{1}),\dots,L^{(j)}(a_{k}))] - \varepsilon k.$$
By Lemma~\ref{lem:cousins-sat},  $f(L^{(j)}(a_{1}),\dots,L^{(j)}(a_{k})) = 1$, if $L^{(j)}(a_{1}),\dots,L^{(j)}(a_k)$ are cousins. Therefore, 
$$
\val(\varphi,\calI_{red})\geq
\min_{
\substack{
a_1,\dots,a_k\in [M]\\
\cousins(a_1,\dots,a_k) = 1}}
\Pr\{\cousins(L^{(j)}(a_{1}),\dots,L^{(j)}(a_{k}))\} - \varepsilon k.
$$
This concludes the proof of Lemma~\ref{lem:L-works}, because $L^{(j)}$ has the same distribution as $L$.
\end{proof}
\section{Random Solutions for Ordinary CSPs on the Gap Instance are Almost Optimal}\label{sec:ordinary-CSPs-gap}
In this section, we prove Lemma~\ref{lem:ordinary-CSPs-gap}. Loosely speaking, this lemma says that every solution to an \emph{ordinary} CSP instance $\calI_{gap}^f$ has value at most $\alpha + \varepsilon$, where $\alpha$ is the optimal \emph{biased} random assignment for this ordinary CSP:
$$\alpha = 
\max_{\rho} \E_{x_i\sim \rho}\big[
f(x_1,\dots,x_k)\big].
$$
See Section~\ref{sec:proof-overview} for details.

We first examine a variant of Lemma~11.3 from the paper by \citet*{GHMRC11}. Instance $\calI_{gap}^f$ is defined on a perfect $k$-ary tree $T$ of depth $d$. Define a probability distribution $\calP$ on the internal nodes of $T$. To draw a random vertex from $\calP$, we first pick a random leaf $u$ of $T$ (with the uniform distribution). We denote the path from the root to $u$ by $u(0),\ldots,u(d-1), u(d)=u$. Then, we pick a random $t$ from $0$ to $d-1$ and output $u(t)$. Note that $u(t)$ has exactly the same distribution as the one we used in the definition of random map $L_{k,m}$ and instance $\calI_{gap}^f$ in Section~\ref{sec:proof-overview}. 

We now consider a  solution $\varphi$ for an \emph{ordinary} CSP with payoff function $f$. Let $\mu_i(T_u)$ be the fraction of leaves in subtree $T_u$ (rooted at $u$) having label $i$ (i.e., leaves $l$ in $T_u$ with $\varphi(l) = i$). Then, the following lemma holds.

\begin{lemma}[cf. Lemma 11.3 in \cite{GHMRC11}]\label{lem:mu-parent-mu-children}
$$\E_{u,t\sim \calP}\Big[\frac{1}{k}\sum_{y\in \child(u(t))}
\sum_{i=1}^q
\big|\mu_i(T_y)-\mu_i(T_{u(t)})\big|\Big]
\leq
\sqrt{\frac{2\log_2 q}{d}}.$$
\end{lemma}

A variant of this lemma was proved by~\cite{GHMRC11}. Their upper bound is a little worse than ours. However, we will only use that the upper bound tends to $0$ as $d$ goes to infinity. We provide a proof of this lemma in Section~\ref{sec:proof-lemma-entropy}. Arguably, our proof is more intuitive than the original proof. However, it is also longer.
We now use Lemma~\ref{lem:mu-parent-mu-children} to prove Lemma~\ref{lem:ordinary-CSPs-gap}.

\begin{proof}[Proof of Lemma~\ref{lem:ordinary-CSPs-gap}] 
Let $\varphi$ be a solution for $\calI_{gap}^f$ where the depth $d$ of the tree $T$ (see above) is sufficiently large. Specifically, $d > (k/\varepsilon)^2\log_2 q$.
The value of $\varphi$ on instance 
$\calI_{gap}^f$
equals 
$$\E[f(
\varphi(L_{k,m}(1)),\dots,
\varphi(L_{k,m}(k))]
$$
because payoff functions in $\calI_{gap}^f$ are defined on $k$-tuples of leaves $(L_{k,m}(1),\dots,L_{k,m}(1))$.
Define another random map $\tilde L_{k,m}$. This function works as $L_{k,m}$ except after choosing a random node $u(t)$, it maps all numbers $1,\dots, k$ randomly into subtree rooted at $u(t)$. For each choice of $u(t)$, the total variation 
distance between the conditional distributions of $L_{k,m}(j)$ and $\tilde L_{k,m}(j))$ equals
$$
\frac{1}{2}\sum_{i=1}^q
\big|\mu_i(T_{u_j(t)})-\mu_i(T_{u(t)})\big|.
$$
Thus, we can couple $\varphi(L_{k,m}(j))$ and $\varphi(\tilde L_{k,m}(j))$ in such a way that (see Lemma~\ref{lem:opt-coupling})
$$\Pr\big(\varphi(L_{k,m}(j)) \neq 
\varphi(\tilde L_{k,m}(j))\mid u(t)\big) = \frac{1}{2}\sum_{i=1}^q
\big|\mu_i(T_{u_j(t)})-\mu_i(T_{u(t)})\big|.
$$
Then,
$$\Pr\big(\exists j\text{ s.t. } \varphi(L_{k,m}(j)) \neq 
\varphi(\tilde L_{k,m}(j))\mid u(t)\big) = 
\frac{1}{2}\sum_{j=1}^k
\sum_{i=1}^q
\big|\mu_i(T_{u_j(t)})-\mu_i(T_{u(t)})\big|.
$$
Finally,
$$\Pr\big(\exists j\text{ s.t. } \varphi(L_{k,m}(j)) \neq \varphi(\tilde L_{k,m}(j))\big) = \frac{1}{2}\E_{u,t\sim\calP}\bigg[
\sum_{i=1}^q
\sum_{j=1}^k
\big|\mu_i(T_{u_j(t)})-\mu_i(T_{u(t)})\big|\bigg].
$$
By Lemma~\ref{lem:mu-parent-mu-children}, the right hand side is upper bounded by $k\sqrt{\dfrac{\log_2 q}{2d}}$.
Thus,
$$\E\big[f(
\varphi(L_{k,m}(1)),\dots,
\varphi(L_{k,m}(k))\big]\leq 
\E\big[f(\varphi(\widetilde L_{k,m}(1)),\dots,
\varphi(\widetilde 
L_{k,m}(k))\big] + k\sqrt{\dfrac{\log_2 q}{2d}}.
$$
Here we used that $f$ is upper bound by $1$. Now, observe that when we use function $\widetilde 
L_{k,m}(k))$, we essentially do a biased random assignment. Namely, we first pick $u(t)$ and then randomly and independently pick labels for $x_1,\dots,x_k$ in the subtree rooted at $u(t)$. It is important that after $u(t)$ is chosen all variables $x_1,\dots, x_k$ are i.i.d. Thus, the first term is upper bounded by $\alpha$. We get
$$\E\big[f(
\varphi(L_{k,m}(1)),\dots,
\varphi(L_{k,m}(k))\big] 
\leq \alpha + k\sqrt{\dfrac{\log_2 q}{2d}}\leq \alpha + \varepsilon.
$$
This concludes the proof of Lemma~\ref{lem:ordinary-CSPs-gap}.
\end{proof}

\section{Generalizations}
\subsection{Phylogenetic CSPs with Multiple Payoff Functions}
\label{sec:multipl-payoff}

We now discuss phylogenetic CSPs with multiple payoff functions $f_1,\dots, f_r$. We assume that they are scaled so that the maximum payoff of each $f_i$ is $1$. First, consider a special case of the problem when the total weight of constraints of every type is prescribed in advance. Namely, suppose that every instance must have $\mu_i$ weight of constraints for payoff function $f_i$ i.e.
$\weight(C_{f_i})=\mu_i$. This variant of the problem is essentially equivalent to the problem with a composite payoff function $f$ defined as follows:
$$
f^*_{\mu}\big(
x^{(1)}_1,\dots, x^{(1)}_k,
x^{(2)}_1,\dots, x^{(2)}_k,
\dots,
x^{(r)}_1,\dots, x^{(r)}_k
\big)
= 
\mu_i \sum_{i=1}^r f_i (x^{(i)}_1,\dots, x^{(i)}_k).
$$
More precisely, the phylogenetic CSP problem with payoff function $f^*_{\mu}$ is a special case of the problem with functions $\{f_i\}$ and prescribed weights $\mu_i$. This is the case simply because $f^*_{\mu}$ can be expressed as the sum of functions $f^{i}$. For every $\mu$, we know the hardness of this problem. It is defined by the approximation threshold
$$\alpha_{opt}(f^*_{\mu}) = 
\sup_{\rho}\alpha_{\rho}(f^*_{\mu}) = 
\sup_{\rho}\sum_{i=1}^r \mu_i\,\alpha_{\rho}(f_i).
$$
Let $\mu^* = \argmin_{\mu}(f^*_{\mu})$. Our phylogenetic problem with functions $f_1,\dots,f_r$ is at least as hard as
$f^*_{\mu^*}$. Consequently, for almost satisfiable instances of phylogenetic CSPs with payoff functions $f_1,\dots, f_r$, it is NP-hard to find a solution of value at least $\alpha_{opt}(f_1,\dots,f_r) + \varepsilon$, where
$$
\alpha_{opt}(f_1,\dots,f_r) = 
\alpha_{opt}(f^*_{\mu^*}) = \sup_{\rho}\sum_{i=1}^r \mu_i\,\alpha_{\rho}(f_i).
$$
Note that approximation 
$\alpha_{opt}(f_1,\dots,f_r)-\varepsilon$
can be achieved. The algorithm can first find the ratios $\mu_i$ and the corresponding distribution $\rho$ (for example, we discretize possible values of $\mu$ and store corresponding $\rho$ in the precomputed table). Furthermore, instead of finding the best $\rho$ for the current weights $\mu$, the algorithm can pick a measure $\rho$ at random from a list of measures. This follows from von Neumann's (\citeyear{minmax}) minimax theorem.

The reader may ask if we can use the same distribution $\rho$ for all instances of phylogenetic CSP $\Gamma$ with several payoff function . It turns out that the answer is no. Consider payoff functions \emph{one split to the left} and \emph{one split right to the right} (see Figure~\ref{fig:split-one-to-left-right}). For every fixed distribution $\rho$, we can find an instance of the problem for which the biased randomized assignment satisfies exponentially small in $k$ fraction of all constraints. However, if we first decide to satisfy only one type of predicates -- \emph{one split to the left} and \emph{one split right to the right} -- and pick the appropriate $\rho$ for it, then we can satisfy $1/2-\varepsilon$ fraction of all constraints.

\subsection{Higher Arity Trees}\label{sec:higher-arity-tree}

In this paper, we proved our main hardness result for \emph{binary} phylogenetic trees. However, the same hardness result also holds for trees of an arbitrary fixed arity $r \geq 2$. To make our proof work for $r$-ary trees, we need to adjust the definitions of the \emph{coarse solution} and bracket predicates, and then slightly modify the proof of Lemma~\ref{lem:good-cs}. Specifically, the coarse solution must satisfy the following conditions:
\begin{enumerate}
\item[1.\ ] (coarse) tree $T$ has at most $q$ leaves;
\item[2.\ ] at most $\varepsilon |V|$ distinct variables have the same color; and
\item[$3'$.] moreover, every color class is the union of at most $2r$ groups of
consecutive variables in ordering $\pi$.
\end{enumerate}
The bracket predicates we need to use for $r$-ary trees have form $[u\to a, v\to b, w\to c]$. This predicate indicates that $u$, $v$, and $w$ must be in subtrees $a$, $b$, and $c$ of the $\lca(u,v,w)$.

Finally, the algorithm from Lemma~\ref{lem:good-cs} should use more than four labels at every step of recursion. When node $u$ is processed, it created $r$ groups of labels, one group for each of $u$'s children. In turn, every group has $r(r-1)$ labels. So, the total number of labels is $r^2(r-1)$.
Suppose that yet unlabeled leaf $l$ belongs to the subtree rooted at the $a$-th child of $u$. Assume that the top processed node in that tree is $v$. Then, $l$ receives label $(a,b,c)$ where $b$ and $c$ are indices of subtrees of $\lca(v,l)$, where
$v$ and $l$ belong to. If there are no processed nodes in the subtree rooted at the $a$-th child of $u$, all leaves in that tree receive label $(a,0,0)$.

\section{Tree Patterns and Bracket Predicates}~\label{sec:cl:bracket-predicates}
In this section, we prove (1) that every phylogenetic payoff function can be defined by a list of pattern and (2) 
every pattern can be expressed as a conjunction of bracket predicates mentioned (Lemma~\ref{lem:bracket-predicates} in Section~\ref{sec:prelim}).

\begin{claim}
Every phylogenetic payoff function can be defined by a list of patterns (with a payoff assigned to each pattern).
\end{claim}
\begin{proof}
Consider a phylogenetic function $f$ of arity $k$. Let $\calP$ be the set of all non-isomorphic irreducible  patterns with $k$ leaves labeled by $x_1,\dots,x_k$. This set is finite because each irreducible tree has $2k-1$ nodes (it is a full binary tree with $k$ leaves). See Section~\ref{sec:prelim} for the definition of homeomorphic trees and reductions. Now for every pattern $P$ with leaves $x_1,\dots, x_k$ in $\calP$,  we compute $f(P,x_1,\dots,x_k)$ (the value of $f$ on pattern $P$)
and assign it to pattern $P$. Finally, we remove all patterns with payoff $0$.

We now prove that the obtained patterns define function $f$. Consider an arbitrary tree $T$ and $k$ leaves $u_1,\dots,u_k$. This tree with with leaves $u_1,\dots,u_k$ can be reduced to some irreducible pattern $P^*$. This pattern $P^*$ with leaves $u_1,\dots,u_k$ and tree $T$ with leaves $u_1,\dots,u_k$ are homeomorphic. Thus,
$f(T, u_1,\dots,u_k) = 
f(P^*, u_1,\dots,u_k)$. Since $P^*$ is irreducible, it must be in the list $\calP$. The value we assign to $P^*$ is $f(P^*, u_1,\dots,u_k)$. Hence, the function defined by the list of pattern obtained above equals $f$.
\end{proof}

To prove Lemma~\ref{lem:bracket-predicates}, we need the following claim.
\begin{claim}\label{cl:diff-nonisomorphic-patterns}
Consider two irreducible non-isomorphic patterns $P_1$ and $P_2$ with $k$ leaves each labeled by $x_1,\dots,x_k$. Then, there exists a bracket predicate such that $P_1$ satisfies this predicate, but $P_2$ does not.
\end{claim}
\begin{proof}
We prove this claim by induction on $k$. For $k=1$, there is only pattern, so $P_1$ must be isomorphic to $P_2$. Suppose $k\geq 2$. Consider the left and right subtrees of $P_1$ and $P_2$: 
$P_1^{left}$,
$P_1^{right}$,
$P_2^{left}$, and
$P_2^{right}$.
Note that each tree $P_1^{left}$,
$P_1^{right}$,
$P_2^{left}$, and
$P_2^{right}$ must be non-empty because $P_1$ and $P_2$ are irreducible. Since $P_1$ and $P_2$ are not isomorphic, one of the two pairs $P_1^{left}$ and
$P_1^{right}$ or 
$P_2^{left}$ and
$P_2^{right}$ must be non-isomorphic. Suppose without loss of generality that 
$P_1^{left}$ and
$P_1^{right}$ are non-isomorphic. Then, we consider two cases. 

\medskip

I. If $P_1^{left}$ contains the same set of leaves as $P_2^{left}$ (e.g. $\{x_3, x_7 , x_8\}$), then we apply the inductive hypothesis to 
$P_1^{left}$ and $P_2^{left}$ and obtain the desired bracket predicate satisfied by $P_1^{left}$ but not $P_2^{left}$. It is also satisfied by $P_1$ but not by $P_2$.

\medskip

II. Suppose now that $P_1^{left}$ and $P_2^{left}$ contain different sets of variables (e.g., $P^{left}_1$ contains $\{x_3, x_7 , x_8\}$ but $P^{left}_2$ contains $\{x_1, x_7 , x_8\}$). If $P_1^{left}$ has a variable $x_i$ which is not in $P_2^{left}$, and $P_2^{left}$ has a variable $x_j$ which is not in $P_1^{left}$, then $P_1$ satisfies $[x_i<x_j]$ but $P_2$ does not. Otherwise, the set of variables in $P_1^{left}$ must be a proper subset of variables in 
$P_2^{left}$ or vice versa. 
Note that if the set of variables in $P_1^{left}$ is a proper subset of variables in 
$P_2^{left}$, then
the set of variables in $P_2^{right}$ is a proper subset of variables in 
$P_2^{right}$. In this case,
let $x_a$ be a common variable in $P_1^{left}$ and $P_2^{left}$, $x_c$ be a common variable in $P_1^{right}$ and $P_2^{right}$, $x_b$ be common variable between $P_1^{right}$ and $P_2^{left}$. We have that $P_1$ satisfies the predicate 
$[x_a < x_b,x_c]$ but $P_2$ does not. The case when the set of variables in $P_2^{left}$ is a proper subset of variables in 
$P_1^{left}$ is handled similarly.
\end{proof}

\noindent\textbf{Lemma~\ref{lem:bracket-predicates}.} \emph{Every pattern can be expressed as a conjunction of bracket predicates.}
\begin{proof}
Let $P$ be a given (ordered, binary) tree pattern on $k$ leaves. We create all bracket constraints $[x_a< x_b]$, $[x_a, x_b < x_c]$, and $[x_a <  x_b, x_c]$ that are satisfied in $P$. We show that the conjunction  of all these predicates define the pattern $P$.

I. If tree $T$ with leaves $u_1,\dots, u_k$ matches pattern $P$ with leaves $x_1,\dots, x_k$, then $T$ must satisfy all generated bracket constraints because
$P$ and $T$ are homeomorphic trees and reductions defined in ~Section~\ref{sec:prelim} preserve the value of every bracket predicate.

II. We now show that if $T$ with leaves $u_1,\dots, u_k$ does not match pattern $P$ with leaves $x_1,\dots, x_k$, then there there is at least one pattern in the description of $P$ that does not match $u_1,\dots, u_k$ in $T$. We reduce $T$ with leaves $u_1,\dots, u_k$ to an irreducible tree $P'$ with leaves $u_1,\dots, u_k$. Leaves $u_1,\dots, u_k$ in this tree or pattern $P'$ satisfy the same set of bracket predicates as in $T$. By Claim~\ref{cl:diff-nonisomorphic-patterns}, there exists a bracket predicate that is satisfied in $P$ but not in $P'$. The same predicate is not satisfied in $T$.

\end{proof}

\section{Example when Uniform Random Assignment Fails}\label{sec:examples}
In Figure~\ref{fig:2-caterpillar-binary-tree}, we provide an example of a phylogenetic predicate of $2k$ variables. If we use a biased random assignment algorithm which assigns variables to the left and right subtrees with fixed probabilities 
$p_{left}$ and $p_{right}$, then we will satisfy an exponentially small in $k$ fraction of all predicates. 

Instead, we should split variables with probability $50\%/50\%$ in the root $r$ of the tree. Then, in each vertex $u$ in the left subtree of $r$, we will assign variables to the left part with probability 
$1-\delta$ and right part with probability $\delta$. We do the opposite in the right subtree of $r$. If $\delta$ is sufficiently small, then the probability that we satisfy this predicate is almost the same as the probability that we split the variables into two equal groups in the root, which equals $\binom{2k}{k}/2^k=\Omega(1/\sqrt{k})$.

\section{Conclusion}

Here we studied a large class of problems that have been studied in various communities that concern how to find hierarchical representation of data, when given as input a collection of local constraints among $n$ data points. Specifically, the input is a set of local information on $k$ items of interest (e.g., species of animals, documents, images etc.) and the goal is to aggregate it into a global hierarchy on the whole dataset of size $n$ that closely agrees with the local information. The most basic case is when the input contains triplet constraints that give information about the relative similarity between 3 points $a,b,c$; such triplet queries are especially useful in crowdsourcing, databases, metric learning, logic, and computational biology. Furthermore, there are various other objectives that have been studied depending on the types of input information that is allowed and/or the properties required of the final hierarchy. Overall, the corresponding problems form a class of constraint satisfaction problems (CSPs) over hierarchies, that are called Phylogenetic CSPs and have been formally studied in the algebraic and logic communities. We note that many of the problems over hierarchies resemble at a high-level analogous formulations of well-motivated problems in the (flat) clustering and ranking literature, e.g., Correlation Clustering, Maximum Acyclic Subgraph, Betweenness etc.

Even though Phylogenetic CSPs have been studied for more than four decades, their approximability was not well-understood. The main result in the paper is that Phylogenetic CSPs are approximation resistant, meaning that they are hard-to-approximate better than a (biased) random assignment. This generalizes previously-known results for ordering CSPs, extends the definition of approximation resistance (to also allow for non-uniform randomized assignments) and it significantly augments the list of approximation resistant predicates by pointing to a large family of hard problems.

\bibliographystyle{abbrvnat}
\bibliography{references.bib}
\appendix
\section{History of the Problems and Further Related Work}

Representing data as a tree is useful across various domains in order to describe the fine-grained relations between items of interest, or to visualize their treelike structure (e.g., in large networks) or the evolutionary history, e.g., for different species in taxonomy, and in natural languages/manuscripts in linguistics. 

The problems considered here are old problems going back to more than four decades ago, to the original work of~\citet*{aho1981inferring} who wanted to understand how to build a hierarchical clustering given ancestry relationships for the leaves. In their paper titled ``Inferring a Tree from Lowest Common Ancestors with an Application to the Optimization of Relational Expression'' the explain how this seemingly unrelated problem of aggregating triplets (triplet reconstruction) has important applications in the area of relational databases. Since then problems finding hierarchical representations on data has been been studied in various communities, as we summarize below:

\begin{itemize}
    \item Databases, Logic and Algebra: ~\cite{aho1981inferring} gave the first algorithm to aggregate triplet constraints that finds a tree that satisfies all of them, if such a tree exists. Interestingly, similar algorithmic ideas were considered by~\cite{steel1992complexity} motivated by applications in computational biology. Generalizations of the Triplet and Quartet Reconstruction problems have been intensively studied in the Computational Logic and Algebraic communities, see for example~\cite{bodirsky2010complexity,bodirsky2012complexity,bodirsky2017complexity} and references therein. Specifically, they study CSPs over trees called Phylogenetic CSPs, which are infinite-domain CSPs and they are interested in the complexity of related problems. Interestingly, there are dichotomy results for Phylogenetic CSPs similar to the dichotomy results observed in complexity of boolean or finite-domain CSPs~\citep*{bulatov2003algebraic,bulatov2006dichotomy,bulatov2017dichotomy}.

    \item Theoretical Computer Science: After the work of~\cite{aho1981inferring}, many works built on improving the runtime of their algorithm using specialized data structures or studying related questions in various settings~\citep*{farach1995agreement,ng1996reconstruction,kannan1998computing,henzinger1999constructing,semple2000supertree}. As we mentioned in the introduction, in terms of approximability not much was known. For the maximization version the best approximation was achieved by the random tree and no progress had been made. Special instances like dense instances were studied in an early work of~\cite{jiang1998orchestrating}, where they gave a PTAS using techniques of~\cite{arora1995polynomial} on instances with $m=\Omega(n^4)$ constraints. Moreover, the work of~\citet*{byrka2010new} studies approximation questions for maximization and minimization variants of triplet reconstruction and the work of~\citet*{brodal2013efficient} gives efficient algorithms for computing distances between trees based on how the two trees differ with respect to triplets. Other methods for constructing trees or comparing trees based on quartets have also been studied in theoretical computer science, for example see the works of~\citet*{alon2014compatibility,alon2016maximum,snir2008quartets,snir2012quartet,snir2012reconstructing}. Finally, the more general CSPs over trees that we studied here  with the constraints involving more than 3 or 4 items, have also been studied as ``subtree/supertree'' aggregation methods~\citep*{jansson2005rooted,dessmark2007polynomial,jansson2012complexity}.
    
    \item Crowdsourcing, Metric Learning and other Machine Learning Applications: Recall, a triplet $ab|c$ indicates that ``$a$ and $b$ are more similar to each other than to $c$''. For example, in Figure~\ref{fig:triplets}, we had $\{\texttt{\{lion, tiger\}|\{tuna\}}\}$. In the context of finding a hierarchy over the dataset, such triplets are interpreted as ``must-link-before'' constraints~\citep*{vikram2016interactive}, which are the analogue of the popular ``must-link'' and ``cannot-link'' constraints that are used in the clustering literature~\citep*{wagstaff2001constrained} (notice that in HC, all points belong in the same cluster initially, and all points are separated at the leaves, so such ``must-link''/``cannot-link'' constraints do not apply). Triplets are especially useful in crowdsourcing and active learning. This is because humans are notoriously bad at providing accurate numerical information, but are quick and precise at comparing items (e.g., answering questions like which pair out of $\{$lion, tiger, tuna$\}$ is most similar); consequently, triplet queries (or more generally ``ordinal'' interactions) have been used to query users for a variety of downstream tasks like tree reconstruction or finding non-metric embeddings (also called ordinal embeddings)~\citep*{schultz2003learning,alon2008ordinal,tamuz2011adaptively,jamieson2011low,awasthi2014local,terada2014local,jain2016finite,emamjomeh2018adaptive}.

    \item Taxonomy and Computational Biology: The study of hierarchical clustering is fundamental in evolutionary biology and the scientific field of Taxonomy tries to uncover the Tree of Life based on the evolutionary relationships among organisms (e.g., by finding similar genetic patterns in their DNA)~\citep*{sneath1973numerical}. Once again, such relationships often take the form of triplets and quartets aggregation methods~\citep*{bandelt1986reconstructing,steel1992complexity,strimmer1996quartet,bryant1997building,semple2000supertree,ng2000difficulty,felsenstein2004inferring}.   
\end{itemize}

\section{Proof of Lemma~\ref{lem:mu-parent-mu-children}}\label{sec:proof-lemma-entropy} 
In this section, we will prove Lemma~\ref{lem:mu-parent-mu-children}
stated in Section~\ref{sec:ordinary-CSPs-gap}.
We will focus on one label $i$. To simplify notation, let us call all leaves having that label red. Let $T_x$ be the subtree of $T$ rooted at $x$. Also, let  $R(T_x)$ and $\mu(T_x)$ be the number of red leaves in $T_x$ and the fraction of red leaves in $T_x$, respectively (for a subtree $T_x$ of depth $d'$, we have $\mu(T_x) = R(T_x)/k^{d'}$). We claim that for a random vertex $u(t)$ (drawn from $\calP$) and each of its children $x$, the number of red leaves in $T_x$ is close to $R(T_{u(t)})/k$ on average. Below, we denote the set of $k$ child nodes of $u(t)$ by $\child(u(t))$.

\begin{lemma}\label{lem:fraction-red-vertices}
For a random internal node $u\sim \calD$, we have
$$\E_{u,t\sim \calP}\Big[\frac{1}{k}\sum_{y\in \child(u(t))}\big|\mu(T_y)-\mu(T_{u(t)})\big|\Big]
\leq
\mu(T)
\sqrt{\frac{2\log_2 \nicefrac{1}{\mu(T)}}{d}}.$$
\end{lemma}
\begin{proof}
We will assume that $T$ has at least one red leaf. Define an auxiliary probability distribution $\calQ$ on the internal nodes of the tree. Pick a random \emph{red} vertex $v$ in $T$. Then, as before, pick an independent $t$ in $\{0,\cdots, d-1\}$ and output $v(t)$ (where $v(0),\cdots,v(d-1), v(d) = v$ is the path from the root of the tree to $v$). Note that in the definition of $\calP$, we pick $u$ uniformly among all leaves of $T$ but in the definition of $\calQ$, we pick $v$ uniformly among all \emph{red} leaves of $T$. Thus, $v(t)=x$ if and only if $v$ is a red leaf in $T_x$ and $t$ is the depth of $x$ in tree $T$. Consequently, 
$$\Pr_{v,t\sim \calQ}\{v(t) = x\} = \frac{R(T_x)}{R(T)} \cdot \frac{1}{d}=
\underbrace{\frac{R(T_x)/k^{d'}}{R(T)/k^d}}_{\mu(T_x)/\mu(T)}
\cdot \Big(\frac{k^{d'}}{k^d}\cdot \frac{1}{d}\Big)
=
\frac{\mu(T_x)}{\mu(T)}\cdot 
\Pr_{u,t\sim \calP}\{u(t) = x\}.
$$
If $\mu(T_x) \neq 0$, then
$$
\Pr_{u,t\sim \calP}\{u(t) = x\} 
= 
\frac{\mu(T)}{\mu(T_x)}\cdot 
\Pr_{v,t\sim \calQ}\{v(t) = x\}.$$
Thus,
\begin{align*}
\E_{u,t\sim \calP}\Big[\sum_{y\in \child(u(t))}\big|\mu(T_y)-\mu(T_{u(t)})\big|\Big]
&=
\E_{v,t\sim \calQ}\Big[\frac{\mu(T)}{\mu(T_{v(t)})}\cdot\sum_{y\in \child(v(t))}\big|\mu(T_y)-\mu(T_{v(t)})\big|\Big]\\
&= \mu(T)\cdot
\E_{v,t\sim \calQ}
\bigg[\sum_{y\in \child(v(t))}\Big|
\frac{\mu(T_y)}{\mu(T_{v(t)})} - 1
\Big|\bigg]\\
&= k\, \mu(T)\cdot
\E_{v,t\sim \calQ}
\bigg[
\sum_{y\in \child(v(t))}
\Big|
\frac{R(T_y)}{R(T_{v(t)})} - \frac{1}{k}
\Big|\bigg].
\end{align*}
In the expectation above, we ignore the terms with $R(T_{v(t)}) = 0$ --- the probability of such $v(t)$ equals $0$.
For an internal node $x$ of $T$, define two distributions, $\calA_x$ and $\calB_x$, on the set of its children $\child(x)$. The first distribution, $\calA_x$, is the uniform distribution on $\child(x)$. The second distribution, $\calB_x$, picks a $y$ in $\child(x)$ with probability proportional to the number of red leaves in $T_y$ i.e., for $y\in \child(x)$,
$$
\Pr_{Y\sim \calB_x}\{Y = y\} = \frac{R(T_y)}{R(T_x)}.
$$
If $T_x$ does not have any red leaves and, consequently, $R(T_x) = 0$, then we let 
$B_x$ be the uniform distribution on $\child(x)$. For $y\in\child(v(t))$, we have
\begin{align*}
\sum_{y\in \child(v(t))}\Big|
\frac{R(T_y)}{R(T_{v(t)})} - \frac{1}{k}
\Big| &= 
\sum_{y\in \child(v(t))}\Big|
\Pr_{Y\sim \calA_{v(t)}}\{Y = y\} 
-
\Pr_{Y\sim \calB_{v(t)}}\{Y = y\} 
\Big|\\
&\leq
2\delta_{TV}(\calA_{v(t)},\calB_{v(t)}),
\end{align*}
where $\delta_{TV}(\calA_{v(t)},\calB_{v(t)})$ is the total variation distance between 
$\calA_{v(t)}$ and $\calB_{v(t)}$.
Thus,
$$\frac{1}{k}\E_{u,t\sim \calP}\Big[\sum_{y\in \child(u(t))}\big|\mu(T_y)-\mu(T_{u(t)})\big|\Big]\leq 2\mu(T)\;\delta_{TV}(\calA_{v(t)},\calB_{v(t)}).
$$

We will now show that the total variation distance between $\calA_{v(t)}$ and $\calB_{v(t)}$ is small on average for a random node $v(t)$ with $v,t\sim Q$. This will conclude the proof of the theorem.
\begin{lemma}
As before, let $\mu(T) = R(T)/k^d$ be the fraction of red leaves in tree $T$. Suppose $\mu(T) > 0$. Then, for a random internal node $v(t)$ having distribution $\calQ$, we have
$$\E_{v,t\sim \calQ}
\Big[\delta_{TV}(\calA_{v(t)}, \calB_{v(t)})\Big]
\leq
\sqrt{\frac{\log_2 \nicefrac{1}{\mu}}{2d}}.
$$
\end{lemma}
\begin{proof}
Let $v$ be a random red vertex in $T$. Random variable $v$ takes $R(T)$ different values with probability $1/R(T)$ each. Hence, its entropy equals
\begin{equation}\label{eq:entropy-v}
H(v) = \log_2 R(T) = 
\log_2(k^d \cdot \mu(T)) =
d \, \log_2 k-\log_2 \nicefrac{1}{\mu}.
\end{equation}
By the chain rule of conditional entropy, we also have
\begin{equation}\label{eq:chain-rule}
H(v) = \sum_{i=0}^{d-1} H(v(i+1) \mid v(i)).
\end{equation}
Observe that the conditional distribution of $v(i+1)$ given $v(i)$ is $\calB_{v(i)}$. Thus,
$$H(v(i+1) \mid v(i))= 
\E_{v} [H(B_{v(i))}].$$
From (\ref{eq:chain-rule}), we have
$$\frac{H(v)}{d} = \frac{1}{d}\sum_{i=0}^{d-1}
\E_{v} [H(B(v(i)))] = 
\E [H(B_{v(t)})],$$
here $t$ is a random number in $\{0,\ldots, d-1\}$ and, consequently, $v(t)$ has distribution $\calQ$.
Using (\ref{eq:entropy-v}), we get
$$\E_{v,t\sim\calQ}
\big[H(B_{v(t)})\big]
= \log_2 k- \frac{\log_2 1/\mu}{d}.$$
We now rearrange the terms and obtain the following bound:
$$\E_{v,t\sim\calQ}\big[\log_2 k - H(B_{v(t)})\big] = \frac{\log_2 1/\mu}{d}.$$
For a fixed $v$ and $t$, random variable $B_{v(t)}$ takes at most $k$ distinct values. Hence, $H(B_{v(t)}) \leq \log_2 k$. Moreover, if $H(B_{v(t)}) = \log_2 k$, then $H(B_{v(t)})$ is uniformly distributed in $\child(v(t))$. That is, $\calB_{v(t)} = \calA_{v(t)}$. Thus, we interpret the expression $\log_2 k - H(B_{v(t)})$ as the distance between $\calB_{v(t)}$ and $\calA_{v(t)}$. In fact, it is equal to the Kullback--Leibler divergence between $\calB_{v(t)}$ and $\calA_{v(t)}$, since
\begin{multline*}
D_{KL}(\calB_{v(t)}\parallel \calA_{v(t)}) = 
-\sum_{y\in \child(v(t))}
\Pr\{B_{v(t)} = y\}\cdot\log_2 \frac{1/k}{\Pr\{B_{v(t)} = y\}} \\
= 
\underbrace{
-\sum_{y\in \child(v(t))}
\Pr\{B_{v(t)} = y\}\cdot\log_2 \frac{1}{k}}_{\log_2 k}
-
\underbrace{
\sum_{y\in \child(v(t))}
\Pr\{B_{v(t)} = y\}\cdot\log_2 \frac{1}{\Pr\{B_{v(t)} = y\}}}_{H(B_{v(t)})}.
\end{multline*}
Therefore,
$\E_{v,t}
\big[D_{KL}(\calB_{v(t)}\parallel \calA_{v(t)})\big] = \frac{\log_2 1/\mu}{d}$.
Finally, by Pinsker's inequality, we have
\begin{multline*}
\E_{v,t\sim\calQ}
\big[\delta_{TV}(\calB_{v(t)}, \calA_{v(t)})]
=
\E_{v,t\sim\calQ}
\Bigg[\sqrt{\frac{D_{KL}(\calB_{v(t)}\parallel \calA_{v(t)})}{2}}\Bigg] \leq \\ \leq
\sqrt{\E_{v,t\sim\calQ}
\bigg[\frac{D_{KL}(\calB_{v(t)}\parallel \calA_{v(t)})}{2}\bigg]}
=
\sqrt{\frac{\log_2 1/\mu}{2d}}.
\end{multline*}
\end{proof}
\end{proof}

Lemma~\ref{lem:mu-parent-mu-children} immediately follows from Lemma~\ref{lem:fraction-red-vertices}:
$$\E_{u,t\sim \calP}\Big[\frac{1}{k}\sum_{y\in \child(u(t))}
\sum_{i=1}^q
\big|\mu_i(T_y)-\mu_i(T_{u(t)})\big|\Big]
\leq
\sum_{i=1}^q
\mu_i(T)
\sqrt{\frac{2\log_2 \nicefrac{1}{\mu_i(T)}}{d}}\leq
\sqrt{\frac{2\log_2 q}{d}}.
$$
The function $t\mapsto t \sqrt{\log_2 \nicefrac{1}{t}}$ is concave and $\nicefrac{1}{q}\sum \mu_i(T) = 1/q$. Thus by Jensen's inequality:
$$
\frac{1}{q}
\sum_{i=1}^q
\mu_i(T)
\sqrt{\log_2 \nicefrac{1}{\mu_i(T)}}\leq
\frac{1}{q}
\sqrt{\log_2 q}.
$$
\section{Triplets to Quartets Reduction}
\label{sec:triplets-to-quartets}
As we shown in the main part of the paper, Triplet Reconstruction~\textsc{MaxTriplets} is hard-to-approximate better than a random assignment, which achieves a $\tfrac13$-approximation.  A very similar situation appears for another basic problem based on arity 4 constraints:

We will need the following simple definition:

\begin{definition}[Quartet]
A quartet $q$, denoted $q=ab|cd$, is an unrooted, unordered, trivalent\footnote{Trivalent is an unrooted tree where every node has degree $3$, except the leaves that have degree $1$.} tree (see Figure~\ref{fig:quartet_example}). tree on 4 leaves $a,b,c,d$ (see Figure~\ref{fig:triplets-to-quartets-fig},~\ref{fig:quartet_example}). An unrooted, unordered, trivalent tree $T$ (containing leaves $a,b,c,d$) is said to be \emph{consistent with $q$} (or $T$ \emph{satisfies} $q$), if the path in $T$ between $a,b$ is disjoint with the path in $T$ between $c,d$. Otherwise, the quartet and the tree are \emph{inconsistent with each other} (or $T$ \emph{violates} $q$). In general, quartets can also have weights $\weight(ab|cd)$.
\end{definition}

The natural optimization problem associated with Quartet Reconstruction is \textsc{MaxQuartets}:
\begin{definition}[\textsc{MaxQuartets} Problem]
Given a set $X$ of $n$ data points and $m$ quartets defined on data points from $X$, find the unrooted, unordered, trivalent tree $T$ that is consistent with as many quartets as possible (per the definition above). 
\end{definition}

We note that in phylogenetics the problem above is called \textit{Unrooted Quartet Consistency}. In general, Quartet methods also have a long history and are widely deployed in computational biology~\citep{bandelt1986reconstructing,strimmer1996quartet,felsenstein2004inferring}. There are other related versions of Quartet Reconstruction (where constraints and the output need to be rooted). All of our hardness results also hold for the rooted quartet reconstruction problem.

\subsection{The Reduction}
Here we present a simple reduction from the rooted triplets consistency problem ({\sc{MaxTriplets}}) to the popular unrooted quartets consistency problem ({\sc{MaxQuartets}}) that has been extensively studied~\citep*{jiang1998orchestrating,alon2014compatibility,snir2008quartets,snir2012reconstructing}. Recall that a triplet $ab|c$ is a rooted tree with $3$ leaves $a,b,c$ and the output is a binary rooted tree, whereas a quartet $ab|c\gamma$ is an unrooted tree with $4$ leaves and the output is an unrooted trivalent tree (every internal node has degree 3).

\begin{claim}\label{cl:triplets-to-quartets}
There is an approximation-preserving reduction from {\sc{MaxTriplets}} to {\sc{MaxQuartets}}.
\end{claim}
\begin{proof}
Given an instance of {\sc{MaxTriplets}} with $m$ triplets $t_1,t_2,\ldots,t_m$ over a set $L$ of $n$ labels, we create an instance of {\sc{MaxQuartets}} with $m$ quartets $q_1,q_2,\ldots,q_m$  over a set $L'$ of $n+1$ labels as follows: 
\begin{itemize}
    \item $L'=L\cup\{\gamma\}$, where $\gamma$ is a distinguished vertex to be used in order to define quartets below.
    \item For every triplet $t_i=a_ib_i|c_i$ of {\sc{MaxTriplets}}, we generate a quartet $q_i=a_ib_i|c_i\gamma$. Notice that $\gamma$ is present in all generated quartets, and $\gamma$ always appears on the side of the ``outsider'' item $c_i$ for each of the triplets $a_i,b_i|c_i$. See Figure~\ref{fig:triplets-to-quartets-fig}.
\end{itemize}
We claim that the generated quartet instance is equivalent to the triplet instance, in the sense that any candidate solution $T$ for triplets (binary rooted tree) can be turned into a candidate solution $T'$ for quartets (trivalent unrooted tree) that satisfies the same number of constraints, and vice versa. 

\begin{figure}[ht]
    \centering
    \includegraphics[scale=1.5]{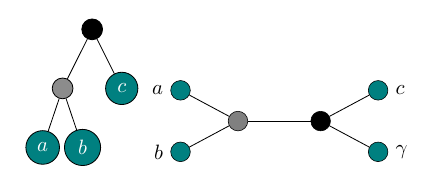}
    \caption{The transformation of a rooted triplet $ab|c$ to an unrooted quartet $ab|c\gamma$ used in the reduction of~Claim~\ref{cl:triplets-to-quartets}.}
    \label{fig:triplets-to-quartets-fig}
\end{figure}

To do so, we start with $T$ and connect its root vertex $r$ (that has degree $2$) to another newly created vertex $\gamma$. Hence the degree of $r$ becomes $3$ and $\gamma$ is a leaf (since its degree is $1$). The final tree corresponds to a trivalent unrooted tree $T'$. Notice that  a triplet $ab|c$ is satisfied by $T$ if and only if the quartet $ab|c\gamma$ is satisfied by $T'$, because the unique path from $a$ to $b$ in $T$ is disjoint from the unique path from $c$ to the root $r$ and hence also to the special vertex $\gamma$. Finally, to turn any unrooted trivalent $T'$ into a binary rooted $T$, we simply root $T'$ at the special vertex $\gamma$. Then, a quartet $ab|c\gamma$ is satisfied by  $T'$ if and only if the triplet $ab|c$ is satisfied by $T$ for the same reason as previously.
\end{proof}

\begin{corollary}
Unrooted Quartets Consistency ({\sc{MaxQuartets}}) is approximation resistant, so it is UGC-hard to beat the (trivial) random assignment algorithm that achieves a $\tfrac13$-approximation.
\end{corollary}
\section{Figures}\label{sec:figures}
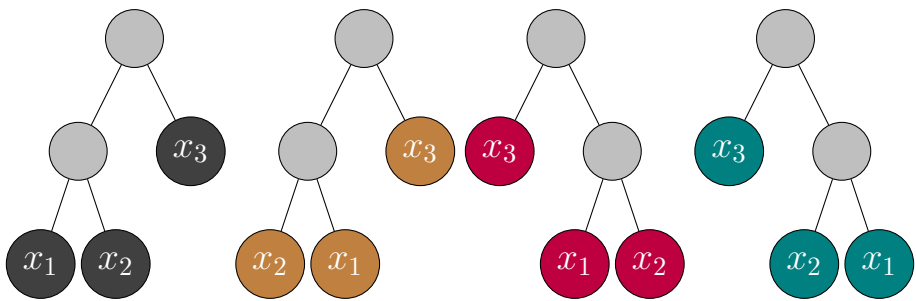
\begin{figure}[ht]
    \centering
\tikzset{
every tree node/.style={minimum width=2em,draw,circle, fill=lightgray},
blank/.style={draw=none},
edge from parent/.style={draw,edge from parent path={(\tikzparentnode) -- (\tikzchildnode)}},
level distance=1.5cm,
}
\begin{tikzpicture}
\Tree
[.\; [.\; 
[.\node[fill=darkgray, text=white]{\Large{$x_1$}}; ] 
[.\node[fill=darkgray, text=white]{\Large{$x_2$}}; ]] [.\node[fill=darkgray, text=white]{\Large{$x_3$}}; ]]
]
\end{tikzpicture}
\begin{tikzpicture}
\Tree
[.\; [.\; 
[.\node[fill=brown, text=white]{\Large{$x_2$}}; ] 
[.\node[fill=brown, text=white]{\Large{$x_1$}}; ]] [.\node[fill=brown, text=white]{\Large{$x_3$}}; ]]
]
\end{tikzpicture}
\begin{tikzpicture}
\Tree
[.\; [.\node[fill=purple, text=white]{\Large{$x_3$}}; ] [.\;
[.\node[fill=purple, text=white]{\Large{$x_1$}}; ] 
[.\node[fill=purple, text=white]{\Large{$x_2$}}; ]]]
]
\end{tikzpicture}
\begin{tikzpicture}
\Tree
[.\; [.\node[fill=teal, text=white]{\Large{$x_3$}}; ] [.\;
[.\node[fill=teal, text=white]{\Large{$x_2$}}; ] 
[.\node[fill=teal, text=white]{\Large{$x_1$}}; ]]]
]
\end{tikzpicture}

\caption{Four patterns that define the Triplets Consistency problem. These pattern can also be specified using the ``square brackets notation''. 
First pattern: $[x_1,x_2< x_3]\;\&\;[x_1< x_2]$. 
Second pattern: $[x_1,x_2< x_3]\;\&\;[x_2<  x_1]$.
Third pattern: $[x_3< x_1,x_2]\;\&\;[x_1 < x_2]$. 
Fourth pattern: $[x_3< x_1,x_2]\;\&\;[x_2 < x_1]$. 
}
\label{fig:patterns-for-triplets}
\end{figure}

\begin{figure}[ht]
    \centering
\tikzset{
every tree node/.style={minimum width=2em,draw,circle, fill=lightgray},
blank/.style={draw=none},
edge from parent/.style={draw,edge from parent path={(\tikzparentnode) -- (\tikzchildnode)}},
level distance=1.5cm,
}
\begin{tikzpicture}
\Tree
[.\node[fill=black, text=white]{\Large{P}}; 
[.\; 
[.\node[fill=teal, text=white]{\Large{$x_1$}}; ] 
[.\node[fill=teal, text=white]{\Large{$x_2$}}; ]] [.\node[fill=teal, text=white]{\Large{$x_3$}}; ]]
]
\end{tikzpicture}
\begin{tikzpicture}
\Tree
[.\node[fill=black, text=white]{\Large{I}}; 
[.\; 
[.\node[fill=teal, text=white]{\Large{$a$}}; ] 
[.\;
[.\node[fill=teal, text=white]{\Large{$b$}}; ]
[.\node[fill=purple, text=white]{\Large{$e$}}; ]
]] 
[.\;
[.\node[fill=teal, text=white]{\Large{$c$}}; ][.\node[fill=purple, text=white]{\Large{$d$}}; ]]]
]
\end{tikzpicture}
\begin{tikzpicture}
\Tree
[.\node[fill=black, text=white]{\Large{II}}; 
[.\; 
[.\node[fill=teal, text=white]{\Large{$b$}}; ] 
[.\node[fill=teal, text=white]{\Large{$a$}}; ]] 
[.\;
[.\node[fill=teal, text=white]{\Large{$c$}}; ][.\node[fill=purple, text=white]{\Large{$d$}}; ]]]
]
\end{tikzpicture}
\begin{tikzpicture}
\Tree
[.\node[fill=black, text=white]{\Large{III}}; 
[.\; 
[.\node[fill=purple, text=white]{\Large{$b$}}; ] 
[.\node[fill=teal, text=white]{\Large{$a$}}; ]] 
[.\;
[.\node[fill=teal, text=white]{\Large{$b$}}; ][.\node[fill=teal, text=white]{\Large{$c$}}; ]]]
]
\end{tikzpicture}

\caption{Consider the leftmost tree $P$ above. It is a pattern on variables $x_1$, $x_2$, $x_3$. Let $f$ be  payoff function defined by this pattern. Namely, let $f(a,b,c) = 1$, if $a$, $b$, $c$ match $P$; $0$, otherwise. Then,
$f(a,b,c)=1$ for the tree I. However, $f(a,b,c)=0$ for tree II, because $a$ and $b$ are  ordered incorrectly. Also, $f(a,b,c)=0$ for tree III, because $a$ is the first node that splits from $a$,$b$, and $c$. 
}
\label{fig:match-no-match}
\end{figure}
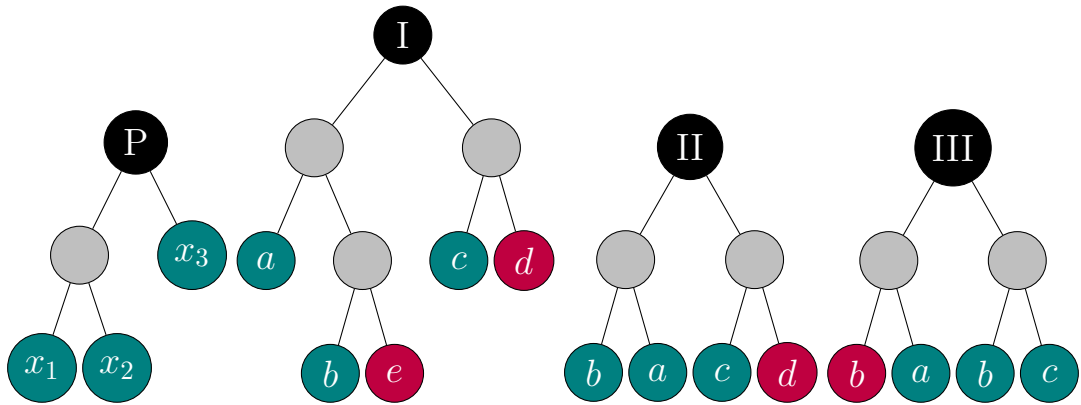

\begin{figure}[ht]
    \centering
\tikzset{
every tree node/.style={minimum width=2em,draw,circle, fill=lightgray},
blank/.style={draw=none},
edge from parent/.style={draw,edge from parent path={(\tikzparentnode) -- (\tikzchildnode)}},
level distance=1.5cm,
}
\begin{tikzpicture}
\Tree
[.\;
[.\; [.\; [.\; [.\; 
[.\node[fill=teal, text=white]{\Large{$x_1$}}; ] 
[.\node[fill=teal, text=white]{\Large{$x_2$}}; ]] [.\node[fill=teal, text=white]{\Large{$x_3$}}; ]]
[.\node[fill=teal, text=white]{\LARGE{$x_4$}};]]
[.\node[fill=teal, text=white]{\Large{$x_5$}};
]]
[.\node[fill=teal, text=white]{\Large{$x_6$}};
]]
\end{tikzpicture}
\begin{tikzpicture}
\Tree
[.\; 
[.\node[fill=purple, text=white]{\Large{$x_1$}}; ] 
[.\; 
[.\node[fill=purple, text=white]{\Large{$x_2$}}; ] 
[.\; 
[.\node[fill=purple, text=white]{\Large{$x_3$}}; ]
[.\; 
[.\node[fill=purple, text=white]{\Large{$x_4$}}; ]
[.\; 
[.\node[fill=purple, text=white]{\Large{$x_5$}}; ]
[.\node[fill=purple, text=white]{\Large{$x_{6}$}}; ]
]]]]]
\end{tikzpicture}

\caption{The left tree is a pattern for the \emph{split-one-to-the-right constraint}. The right tree is a pattern for the \emph{split-one-to-the-left constraint}. Each of the constraints contains all $6!$ permutations of variables $x_1,\dots,x_6$. So, the order in which variables split from others is not important.}
\label{fig:split-one-to-left-right}
\end{figure}
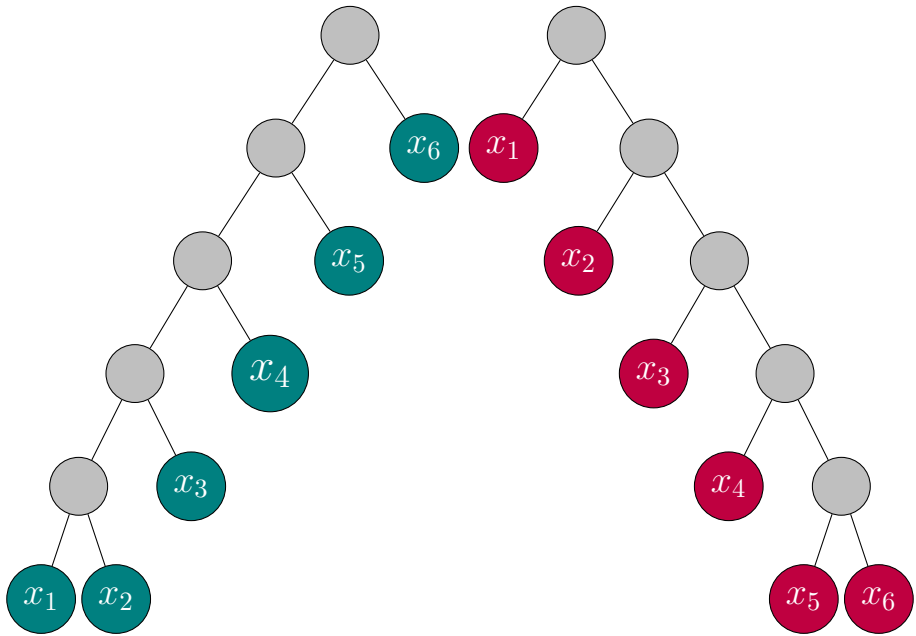

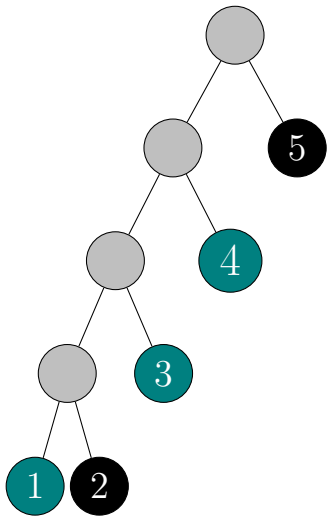
\begin{figure}[ht]
    \centering
\tikzset{
every tree node/.style={minimum width=2em,draw,circle, fill=lightgray},
blank/.style={draw=none},
edge from parent/.style={draw,edge from parent path={(\tikzparentnode) -- (\tikzchildnode)}},
level distance=1.5cm,
}
\begin{tikzpicture}
\Tree
[.\; [.\; [.\; [.\; 
[.\node[fill=teal, text=white]{\Large{$1$}}; ] 
[.\node[fill=black, text=white]{\Large{$2$}}; ]] [.\node[fill=teal, text=white]{\Large{$3$}}; ]]
[.\node[fill=teal, text=white]{\LARGE{$4$}};]]
[.\node[fill=black, text=white]{\Large{$5$}};
]]
\end{tikzpicture}
\tikzset{
every tree node/.style={minimum width=2em,draw,circle, fill=lightgray},
blank/.style={draw=none},
edge from parent/.style={draw,edge from parent path={(\tikzparentnode) -- (\tikzchildnode)}},
level distance=1.5cm,
}

\caption{Binary left caterpillar with five leaves. The right child of each internal node is a leaf. Observe that $\triplet(a,b,c) = 1$ if $a< b < c$. For example, 
$\triplet(1,3,4)=1$.}
\label{fig:caterpillar-binary-tree}
\end{figure}

\begin{figure}[ht]
    \centering
\tikzset{
every tree node/.style={minimum width=2em,draw,circle, fill=lightgray},
blank/.style={draw=none},
edge from parent/.style={draw,edge from parent path={(\tikzparentnode) -- (\tikzchildnode)}},
level distance=1.5cm,
every level 4 node/.style={minimum width=2em,draw,circle, fill=cyan, text=black},
}

\begin{tikzpicture}
\Tree
[.\node[fill=darkgray, text=white]{\LARGE{$u$}};
    [.\LARGE{$x$}
[.\node[fill=darkgray, text=white]{\LARGE{$u_1$}};
       [.\; ] 
       [.\; ]]
       [.\node[fill=darkgray, text=white]{\LARGE{$u_2$}};
    [.\; ]
    [.\; ]]]
    [.\node[fill=darkgray, text=white]{\Large{$w/u_3$}};
    [.\LARGE{$y$}
    [.\node[fill=darkgray, text=white]{\LARGE{$w_1$}}; ]
    [.\node[fill=darkgray, text=white]{\LARGE{$w_2$}}; ]
    ]
    [.\node[fill=darkgray, text=white]{\LARGE{$w_3$}};
    [.\; ]
    [.\; ]
         ]
    ]
]
\end{tikzpicture}
\caption{Binary tree $T'$ constructed based on ternary tree $T$. Nodes $u_1,u_2,u_3$ are children of $u$ in  ternary tree $T$. They are leaves in the pattern tree that consists of vertices $u$, $x$, $u_1$, $u_2$, and $u_3$. Similarly, vertices $w_1,w_2,w_3$ are children of $w$ in $T$. They are leaves in the pattern tree that consists of vertices $w=u_3$, $y$, $w_1$, $w_2$, and $w_3$.}
\label{fig:satisfiable-CSP}
\end{figure}
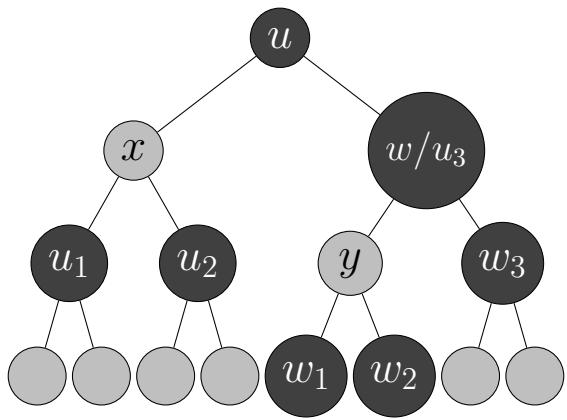

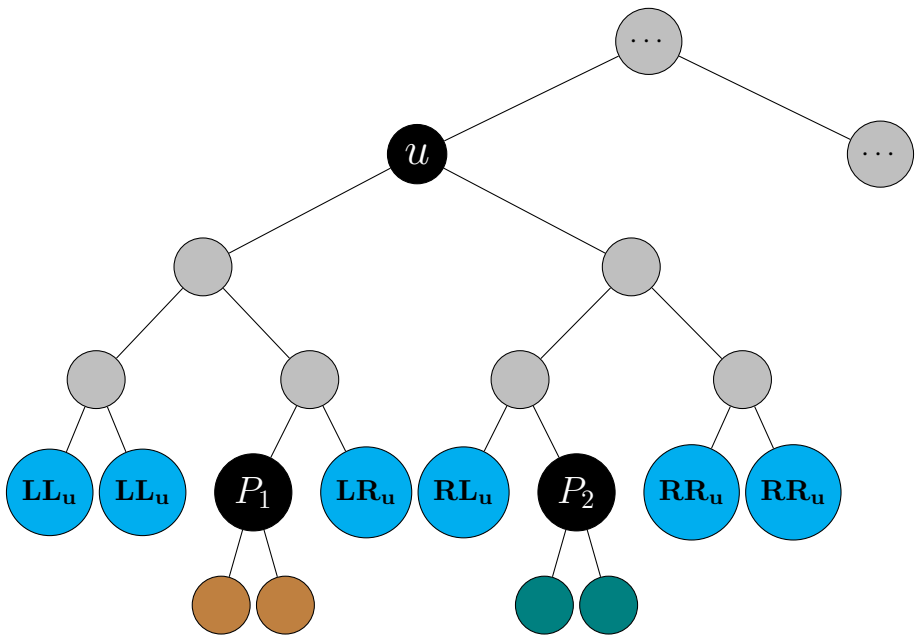
\begin{figure}[ht]
    \centering
\tikzset{
every tree node/.style={minimum width=2em,draw,circle, fill=lightgray},
blank/.style={draw=none},
edge from parent/.style={draw,edge from parent path={(\tikzparentnode) -- (\tikzchildnode)}},
level distance=1.5cm,
every level 4 node/.style={minimum width=2em,draw,circle, fill=cyan, text=black},
}
\begin{tikzpicture}
\Tree
[.$\cdots$
[.\node[fill=black, text=white]{\LARGE{$u$}};
    [.\; [.\; 
       [.$\mathbf{LL_u}$ ] 
       [.$\mathbf{LL_u}$ ]][.\; 
    [.\node[fill=black, text=white]{\Large{$P_1$}};
        [.\node[fill=brown]{}; ]
        [.\node[fill=brown]{}; ]
    ]
    [.$\mathbf{LR_u}$ ]]]
    [.\;
    [.\; 
    [.$\mathbf{RL_u}$ ]
    [.\node[fill=black, text=white]{\Large{$P_2$}};
    [.\node[fill=teal ]{}; ]
    [.\node[fill=teal ]{}; ]
    ]    
    ]
    [.\; [.{$\mathbf{RR_u}$} ]
         [.{$\mathbf{RR_u}$} ]
         ]
    ]
]
[.$\cdots$ ]
]
\end{tikzpicture}
\caption{Algorithm for constructing a coarse solution.
Vertices $P_1$ and $P_2$ are already processed by the algorithm. The algorithm is currently processing vertex $u$.
It assigns four labels $LL_u$, $LR_u$, $RL_u$, $RR_u$ to yet unlabeled leaves in subtree rooted at $u$.}
\label{fig:coloring-algorithm}
\end{figure}

\begin{figure}[ht]
    \centering
\tikzset{
every tree node/.style={minimum width=2em,draw,circle, fill=lightgray},
blank/.style={draw=none},
edge from parent/.style={draw,edge from parent path={(\tikzparentnode) -- (\tikzchildnode)}},
level distance=1.5cm,
}
\begin{tikzpicture}
\Tree
[.\;
[.\; [.\; [.\; [.\; 
[.\node[fill=teal, text=white]{\Large{$x_1$}}; ] 
[.\node[fill=teal, text=white]{\Large{$x_2$}}; ]] [.\node[fill=teal, text=white]{\Large{$x_3$}}; ]]
[.\node[fill=teal, text=white]{\LARGE{$x_4$}};]]
[.\node[fill=teal, text=white]{\Large{$x_5$}};
]]
[.\; 
[.\node[fill=purple, text=white]{\Large{$x_6$}}; ] 
[.\; 
[.\node[fill=purple, text=white]{\Large{$x_7$}}; ]
[.\; 
[.\node[fill=purple, text=white]{\Large{$x_8$}}; ]
[.\; 
[.\node[fill=purple, text=white]{\Large{$x_9$}}; ]
[.\node[fill=purple, text=white]{\Large{$x_{10}$}}; ]
]]]]]
\end{tikzpicture}
\caption{This phylogenetic predicate consists of patterns obtained from the pattern above by permuting variables $x_1,\dots,x_{10}$. The  predicate requires that at some node $u$ variables 
$x_1,\dots,x_{10}$ are split into two equal groups. The first group is assigned to the left subtree; the second group is assigned to the right subtree. Then, the variables in the first group should satisfy the \emph{split-one-to-the-right} constraint, and 
variables in the second group should satisfy the \emph{split-one-to-the-left} constraint (see~Figure~\ref{fig:split-one-to-left-right}).}.
\label{fig:2-caterpillar-binary-tree}
\end{figure}
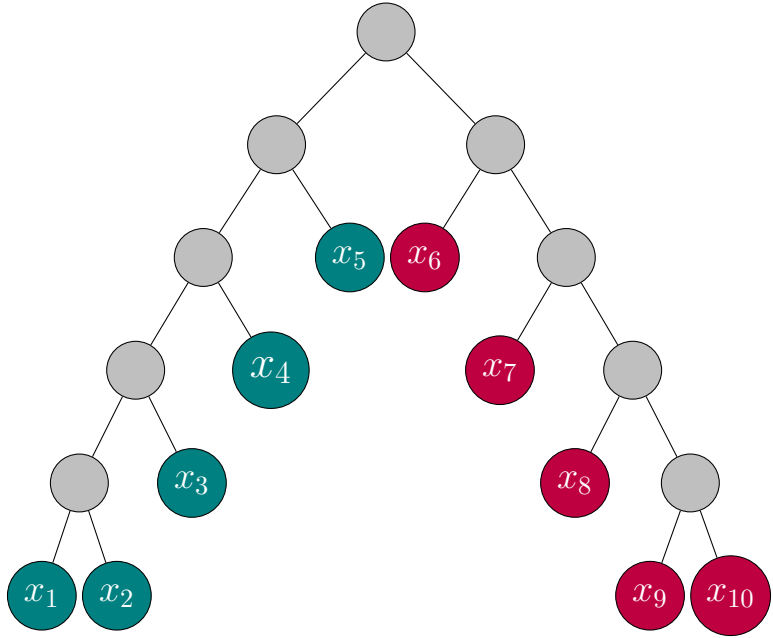

\begin{figure}[ht]
    \centering
    \includegraphics[scale=1.4]{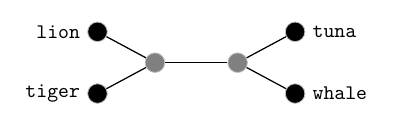}
    \caption{A quartet tree is the smallest
informative unrooted tree used in phylogenetic reconstruction~(\cite{felsenstein2004inferring,snir2008quartets}). Here the quartet $\{\{\texttt{lion, tiger}\},  \{\texttt{tuna, whale}\}\}$ is shown.}
    \label{fig:quartet_example}
\end{figure}

\begin{figure}[ht]
    \centering
    \includegraphics[scale=1.13]{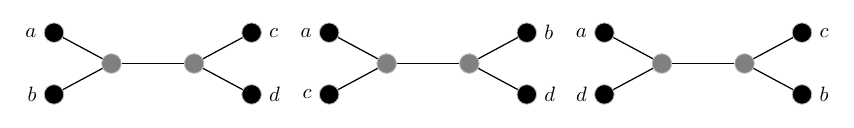}
    \caption{There are only $3$ different (unrooted) quartet trees for items $a,b,c,d$. The performance of a random assignment achieves a $\tfrac13$-approximation, in expectation.}
    \label{fig:quartet_example_2}
\end{figure}

\end{document}